\documentclass[11pt]{article}

\usepackage[margin=1in]{geometry}
\usepackage[utf8]{inputenc}
\usepackage[english]{babel}

\usepackage{microtype}
\usepackage[pdftex]{graphicx}

\usepackage[cmex10]{amsmath}
\usepackage{todonotes}
\usepackage{amsthm}
\usepackage{amssymb}
\usepackage{xspace}
\usepackage{cite}
\usepackage{comment}
\usepackage{letltxmacro}
\usepackage{enumerate}
\usepackage{url}
\usepackage{hyperref}
\usepackage{textcomp}
\usepackage{wrapfig}
\usepackage{float}
\usepackage{caption}
\usepackage{subcaption}
\usepackage{tikz}

\newtheorem{theorem}{Theorem}
\newtheorem{lemma}[theorem]{Lemma}
\newtheorem{claim}[theorem]{Claim}
\newtheorem{corollary}[theorem]{Corollary}
\newtheorem{proposition}[theorem]{Proposition}

\theoremstyle{definition}
\newtheorem{definition}[theorem]{Definition}

\newcommand{\poly}{\mathrm{poly}}
\newcommand{\Oh}{\ensuremath{\mathcal{O}}}

\newcommand{\NN}{\ensuremath{\mathbb{N}}}

\def\T{\ensuremath{\mathcal{T}}}
\def\F{\ensuremath{\mathcal{F}}\xspace}
\newcommand{\tw}{\ensuremath{\mathsf{tw}}}
\newcommand{\tctw}{\ensuremath{\mathsf{tctw}}}
\newcommand{\LCA}{\mathsf{lca}}
\newcommand{\OPT}{\mathsf{OPT}}
\newcommand{\splitterF}{\mathfrak{F}}
\newcommand{\adh}{\ensuremath{\mathsf{adh}}}
\newcommand{\width}{\ensuremath{\mathsf{width}}}
\newcommand{\dist}{\ensuremath{\mathsf{dist}}}
\newcommand{\polylog}{\text{\rm poly\!}\log}
\newcommand{\depth}{\ensuremath{\mathsf{depth}}}

\newcommand{\MX}{\mathsf{MAX}_{\F}}
\newcommand{\G}{\mathcal{G}}
\newcommand{\Rr}{\mathcal{R}}
\newcommand{\Ss}{\mathcal{S}}
\newcommand{\Ii}{\mathcal{I}}
\newcommand{\extnd}[1]{\mathsf{ext}(#1)}
\newcommand{\cop}[1]{\mathsf{copy}(#1)}

\def\cqedsymbol{\ifmmode$\lrcorner$\else{\unskip\nobreak\hfil
\penalty50\hskip1em\null\nobreak\hfil$\lrcorner$
\parfillskip=0pt\finalhyphendemerits=0\endgraf}\fi} 

\newcommand{\cqed}{\renewcommand{\qed}{\cqedsymbol}}

\newcommand{\svg}[2]{\def\svgwidth{#1}{\input{#2.pdf_tex}}}

\hypersetup{
	colorlinks,
	linkcolor={red!60!black},
	citecolor={green!50!black},
	urlcolor={blue!70!black},
	pdfdisplaydoctitle=true,
}

\newcommand{\defparproblem}[4]{
	\vspace{1mm}
	\noindent\fbox{
		\begin{minipage}{0.95\textwidth}
			#1 \\
			{\bf{Input:}} #2  \\
			{\bf{Parameter:}} #3  \\			
			{\bf{Question:}} #4
		\end{minipage}
	}
}

\title{Linear kernels for edge deletion problems\\ to immersion-closed graph classes}

\newcommand{\septhanks}{\ $^,$}

\makeatletter\let\@fnsymbol\@alph\makeatother

\author{Archontia C. Giannopoulou\thanks{Technische Universit\"{a}t Berlin, Berlin, Germany.}\septhanks\thanks{The research of this author has been supported by the European Research Council (ERC) under the European Union’s Horizon 2020 research and innovation programme (ERC consolidator grant DISTRUCT, agreement No 648527) and by the Warsaw Center of Mathematics and Computer Science.} 
\and 
Michał Pilipczuk\thanks{Institute of Informatics, University of Warsaw, Poland.}\septhanks\thanks{Supported by the Polish National Science Center grant SONATA DEC-2013/11/D/ST6/03073.}
\and 
Jean-Florent Raymond$^{\text{c},}$\thanks{AlGCo project team, CNRS, LIRMM, Montpellier,
    France.}\septhanks\thanks{Supported by the Polish
    National Science Centre grant PRELUDIUM DEC-2013/11/N/ST6/02706.}
\and Dimitrios M. Thilikos$^{\text{e,}}$\thanks{Department of Mathematics, National and Kapodistrian University of Athens, Greece.}
\and 
Marcin Wrochna$^{\text{c,d}}$
}

\begin{document}
\maketitle

\begin{abstract}
\noindent Suppose $\F$ is a finite family of graphs. 
We consider the following meta-problem, called {\sc{$\F$-Immersion Deletion}}: 
given a graph $G$ and integer $k$, decide whether the deletion of at most $k$ edges of $G$  can result in a graph that does not contain any graph from $\F$ as an immersion.
This problem is a close relative of the {\sc{$\F$-Minor Deletion}} problem studied by Fomin et al.~[FOCS 2012], where one deletes vertices in order to remove all minor models of graphs from $\F$.
We prove that whenever all graphs from $\F$ are connected  and at least one graph of $\F$ is planar and subcubic, then the  {\sc{$\F$-Immersion Deletion}} problem admits:
\begin{itemize} 
\item a constant-factor approximation algorithm running in time $\Oh(m^3  \cdot n^3 \cdot \log m)$;
\item a linear kernel that can be computed in time $\Oh(m^4 \cdot n^3 \cdot \log m)$; and
\item a $\Oh(2^{\Oh(k)} + m^4 \cdot n^3 \cdot \log m)$-time fixed-parameter algorithm,
\end{itemize}
where $n,m$ count the vertices and edges of the input graph.
These results mirror the findings of Fomin et al.~[FOCS 2012], who obtained a similar set of algorithmic results for {\sc{$\F$-Minor Deletion}},  
under the assumption that at least one graph from $\F$ is planar.
An important difference is that we are able to obtain a {\sl linear} kernel for {\sc{$\F$-Immersion Deletion}}, 
while the exponent of the kernel of Fomin et al. for {\sc{$\F$-Minor Deletion}}  depends heavily on the family $\F$.
In fact, this dependence is unavoidable under plausible complexity assumptions, as proven by Giannopoulou et al.~[ICALP 2015].
This reveals that the kernelization complexity of {\sc{$\F$-Immersion Deletion}} is quite different than that of {\sc{$\F$-Minor Deletion}}.

\end{abstract}

\section{Introduction}\label{sec:intro}
\paragraph*{On the {\sc{$\F$-Minor Deletion}} problem.} 
Given an class of graphs ${\cal G}$, we denote by ${\bf obs}_{\rm mn}({\cal G})$ the {\em minor-obstruction set} of ${\cal G}$, that is  
the set of minor-minimal graphs that do not belong in ${\cal G}$. 
Let us fix some finite family of graphs $\F$. 
A graph $G$ is called {\em{$\F$-minor-free}} if $G$ does not contain any graph from $\F$ as a minor.
The celebrated Graph Minors Theorem of Robertson and Seymour~\cite{RobertsonS04} implies that for 
every family of graphs $\Pi$ that is  closed under taking minors, the set $\F_{\cal G}={\bf obs}_{\rm mn}({\cal G})$
is finite. In other words, ${\cal G}$ is characterized by the minor-exclusion of finite set of graphs; 
that is ${\cal G}$ is exactly the class of $\F_{\cal G}$-minor-free graphs.
Hence, studying the classes of $\F$-minor-free graphs for finite families $\F$ is the same as studying general minor-closed properties of graphs.

Fomin et al.~\cite{FominLMS12} performed an in-depth study of the following parameterized\footnote{A {\em parameterized} problem can be seen as a subset of $\Sigma^*\times\mathbb{N}$ where its instances are pairs $(x,k)\in \Sigma^*\times\mathbb{N}$. For graph problems, the string $x$ usually encodes a graph $G$. 
A parameterized problem admits an {\sf FPT}-algorithm, or, equivalently, belongs in the parameterized complexity class {\sf FPT}, if it can be solved by an $f(k)\cdot {|x|}^{O(1)}$ step algorithm. See~\cite{FlumGrohebook,platypus,Niedermeierbook06} for more on parameterized algorithms and complexity.} problem, named {\sc{$\F$-Minor Deletion}}%
\footnote{Fomin et al. use the name {\sc{$\F$-Deletion}}, but we choose to use the word ``minor'' explicitly to distinguish it from immersion-related problems that we consider in this paper.}:
{\sl Given a graph $G$ and an integer parameter $k$, decide whether it is possible to remove at most $k$ vertices from $G$ to obtain an $\F$-minor-free graph}.
By considering different families $\F$, the {\sc{$\F$-Minor Deletion}} problem generalizes a number of concrete problems of prime importance in parameterized complexity, such as
{\sc{Vertex Cover}}, {\sc{Feedback Vertex Set}}, or {\sc{Planarization}}.
It is easy to see that, for every fixed $k$, the graph class ${\cal G}_{k,{\cal F}}^{\rm mn}$, consisting of the graphs in the YES-instances $(G,k)$ of {\sc{$\F$-Minor Deletion}}, is closed under taking of minors. 
We define ${\cal O}_{k}^{\rm mn}={\bf obs}_{\rm mn}({\cal G}_{k,{\cal F}}^{\rm mn})$.
By the fact that ${\cal O}_{k}^{\rm mn}$  is finite 
and the meta-algorithmic consequences 
of the Graph Minors series of Robertson and Seymour~\cite{RobertsonS04,RobertsonRXIII}, it follows (non-constructively) that  {\sc{$\F$-Minor Deletion}}  admits an {\sf FPT}-algorithm. The optimization of 
the running time of such {\sf FPT}-algorithms for several instantiations of ${\cal F}$ has been an interesting  project
in parameterized algorithm design and so far it has been focused on problems generated by 
minor-closed graph classes.

The goal of Fomin et al.~\cite{FominLMS12} was to obtain results of general nature for {\sc{$\F$-Minor Deletion}}, which would explain why many concrete problems captured as its subcases are 
efficiently solvable using parameterized algorithms and kernelization. This has been achieved under the assumption that $\F$ contains at least one planar graph. 
More precisely, for any class $\F$ that contains at least one planar graph, the work of Fomin et al.~\cite{FominLMS12} gives the following:
\begin{enumerate}[(i)]
\item a randomized constant-factor approximation running in time $\Oh(nm)$;
\item a polynomial kernel for the problem; that is, a polynomial-time algorithm that, given an instance $(G,k)$ of {\sc{$\F$-Minor Deletion}}, 
      outputs an equivalent instance $(G',k')$ with $k'\leq k$ and $|G'|\leq \Oh(k^c)$, for some constant $c$ that depends on~$\F$;
\item an {\sf FPT}-algorithm solving {\sc{$\F$-Minor Deletion}} in time $2^{\Oh(k)}\cdot n^2$.
\item a proof that every graph in ${\cal O}_{k}^{\rm mn}$  has $k^{c_{{\cal F}}}$ vertices for some constant $c_{\cal F}$ that depends (non-constructively) on ${\cal F}$.
\end{enumerate}
We remark that, for the {\sf FPT}-algorithm, the original paper of Fomin et al.~\cite{FominLMS12} needs one more technical assumption, namely that all the graphs from $\F$ are connected.
The fact that this condition can be lifted was proved in a subsequent work of Kim et al.~\cite{KimLPRRSS13}.

The assumption that $\F$ contains at least one planar graph is crucial for the approach of Fomin et al.~\cite{FominLMS12}.
Namely, from the Excluded Grid Minor Theorem of Robertson and Seymour~\cite{RobertsonS86} it follows that for such families $\F$, 
 $\F$-minor-free graphs have treewidth bounded by a constant depending only of~$\F$.
Therefore, a YES-instance of {\sc{$\F$-Minor Deletion}} roughly has to look like a constant-treewidth graph plus $k$ additional vertices that can have arbitrary connections.
Having exposed this structure, Fomin et al.~\cite{FominLMS12} apply protrusion-based techniques that originate in the work on {\em{meta-kernelization}}~\cite{BodlaenderFLPST09,FominLST10}. 
Roughly speaking, the idea is to identify large parts of the graphs that have constant treewidth and a small interface towards the rest of the graph (so-called {\em{protrusions}}), 
which can be replaced by smaller gadgets with the same combinatorial behaviour. Such preprocessing based on {\em{protrusion replacement}} is the base of all three aforementioned results for {\sc{$\F$-Minor Deletion}}.
In the absence of a constant bound on the treewidth of an $\F$-minor-free graph, the technique breaks completely. 
In fact, the kernelization complexity of {\sc{Planarization}}, that is, {\sc{$\F$-Minor Deletion}} for $\F=\{K_5,K_{3,3}\}$, is a notorious open problem. 

An interesting aspect of the work of Fomin et al.~\cite{FominLMS12} is that the exponent of the polynomial bound on the size of the kernel for {\sc{$\F$-Minor Deletion}} grows quite rapidly with the family $\F$.
Recently, it has been shown by Giannopoulou et al.~\cite{GiannopoulouJLS15} that in general this growth is probably unavoidable: 
For every constant $\eta$, the {\sc{Treewidth-$\eta$ Deletion}} problem (delete $k$ vertices to obtain a graph of treewidth at most $\eta$) has no kernel with $\Oh(k^{\eta/4-\epsilon})$ vertices for any $\epsilon>0$, 
unless $\mathsf{NP}\subseteq \mathsf{coNP}/\mathsf{poly}$. Since graphs of treewidth $\eta$ can be characterized by a finite set of forbidden minors $\F_\eta$, at least one of which is planar,
this refutes the hypothesis that all {\sc{$\F$-Minor Deletion}} problems admit polynomial kernels with a uniform bound on the degree of the polynomial. 
However, as shown by Giannopoulou et al.~\cite{GiannopoulouJLS15}, such {\em{uniform kernelization}} can be achieved for some specific problem families, like vertex deletion to graphs of constant tree-depth.

\paragraph*{Immersion problems.}
Recall that a graph $H$ can be {\em{immersed}} 
into a graph $G$ (or that $H$ is an {\em immersion} of $G$) if there is a mapping from $H$ to $G$ that maps vertices of $H$ to pairwise different vertices of $G$ and edges of $H$ to pairwise
{\sl edge-disjoint paths} connecting the images of the respective endpoints\footnote{In this paper we consider {\em{weak immersions}} only, as opposed to {\em{strong immersions}} 
where the paths are forbidden to traverse images of vertices other than the endpoints of the corresponding edge.}. 
Such a mapping is called an {\em{immersion model}}.
Just like the minor relation, the immersion relation imposes a partial order on the class of graphs.
Alongside with the minor order, Robertson and Seymour~\cite{RobertsonS10} proved that  graphs
are also well-quasi-ordered under the immersion order, i.e., every set of graphs that are 
pair-wise non-comparable with respect to the immersion relation is finite.
This implies that for every  class of graphs ${\cal G}$ that is closed under taking immersions 
the set ${\bf obs}_{\rm im}({\cal G})$, containing the immersion minimal graphs that do not belong in ${\cal G}$, is finite (we call ${\bf obs}_{\rm im}{\cal G}$ {\em immersion obstruction set} of ${\cal G}$). Therefore
${\cal G}$ can be characterized by a finite set of forbidden immersions.
The general intuition is that immersion is a containment relation on graphs that corresponds to {\sl edge} cuts, whereas the minor relation corresponds to {\sl vertex} cuts.
Also, the natural setting for immersions is the setting of {\em{multigraphs}}. Hence, from now on all the graphs considered in this paper may have parallel edges connecting the same pair of endpoints. 

Recently, there has been a growing interest in immersion-related problems~\cite{MarxW14,DvorakW15astru,Ganian0S15,KimOPST15,GiannopoulouKRT16pack,Giannopoulou2014effe,Wollan15,BoothGLR1999,Devos2014,BelmonteGLT2016thes,GovindanR01awea,DvorakY15comp} both from the combinatorial and the algorithmic point of view.
Most importantly for us, Wollan proved in~\cite{Wollan15} an analog of the Excluded Grid Minor Theorem, 
which relates the size of the largest wall graph that is contained in a graph as an immersion 
with a new graph parameter called {\em{tree-cut width}}.
By a {\em{subcubic graph}} we mean a graph of maximum degree at most $3$. 
The following theorem  follows from 
the work of Wollan~\cite{Wollan15} and 
summarizes the conclusions of this work that are important for us. 

\begin{theorem}[\!\!\cite{GiannopoulouKRT16pack}]\label{thm:woland}
For every graph $H$ that is planar and subcubic there exists a constant $a_H$, such that every graph
that does not contain  $H$  as an immersion  has tree-cut width bounded by $a_H$. 
\end{theorem}

In other words, for any family $\F$ of graphs that contains some planar subcubic graph, the tree-cut width of $\F$-immersion-free graphs is bounded by a universal constant depending on $\F$ only.
In Section~\ref{sec:prelims} we discuss the precise definition of tree-cut width and how exactly Theorem~\ref{thm:woland} follows from the work of Wollan~\cite{Wollan15}.
Also, note that if a family of graphs $\F$ does not contain any planar subcubic graph, then there is no uniform bound on the tree-cut width of $\F$-immersion-free graphs.
Indeed, wall graphs are then $\F$-immersion-free, because all their immersions are planar and subcubic, and they have unbounded tree-cut width.

After the introduction of tree-cut width by Wollan~\cite{Wollan15}, the new parameter gathered substantial interest from the algorithmic and combinatorial community~\cite{Ganian0S15,MarxW14,GiannopoulouKRT16pack,KimOPST15}.
It seems that tree-cut width serves  the same role for immersion-related problems as treewidth serves for minor-related problems and, in a sense, it can be seen as an ``edge-analog'' of treewidth.
In particular, given the tree-cut width bound of Theorem~\ref{thm:woland} and the general approach of Fomin et al.~\cite{FominLMS12} to {\sc{$\F$-Minor Deletion}},
it is natural to ask whether the same kind of results can be obtained for immersions where
the considered modification is edge removal instead of vertex removal.
More precisely, fix a finite family of graphs $\F$ containing some planar subcubic graph and consider the following {\sc{$\F$-Immersion Deletion}} problem: given a graph $G$ and an integer $k$, 
determine whether it is possible to delete at most $k$ edges
 of $G$ in order to obtain a graph that does not admit any graph from $\F$ as an immersion.
 
Parallel to the  case of {\sc{$\F$-Minor Deletion}},
for every fixed $k$, the graph class ${\cal G}^{\rm im}_{k,{\cal F}}$ consisting of the graphs in the YES-instances $(G,k)$ of {\sc{$\F$-Immersion Deletion}} is closed under taking of immersions\footnote{Notice that if we  
consider deletion of vertices instead of edges, then the graph class ${\cal G}_{k}^{\rm im}$  is {\sl not} closed under taking immersions (for example, in a star on $7$ vertices with duplicated edges, deleting one vertex makes it $K_3$-immersion-free, but this `duplicated' star immerses $2K_3$, which has no such vertex). This is the main reason why we believe that \emph{edge} deletion gives a more suitable counterpart to {\sc{$\F$-Minor Deletion}} for the case of immersions.}, therefore ${\cal O}_{k}^{\rm im}={\bf obs}_{\rm im}({\cal G}^{\rm im}_{k,{\cal F}})$
is a finite set, by the well-quasi-ordering 
of graphs under immersions~\cite{RobertsonS10}. Together with the immersion-testing algorithm of Grohe et al.~\cite{GroheKMW11find}, this implies that {\sc{$\F$-Immersion Deletion}} admits (non-constructively) 
an {\sf FPT}-algorithm. This naturally  induces the parallel project of optimizing the performance
of such {\sf FPT}-algorithms
for several instantiations of ${\cal F}$. 
More concretely, is it possible to extend the general framework of Fomin et al.~\cite{FominLMS12} to obtain efficient approximation, kernelization, and {\sf FPT} algorithms also for {\sc{$\F$-Immersion Deletion}}?
Theorem~\ref{thm:woland} suggests that the suitable analog of the assumption from the minor setting that $\F$ contains a planar graph 
should be the assumption that at least one graph from $\F$ is planar and subcubic.

\paragraph*{Our results.}
In this work we give a definitive positive answer to this question. The following two theorems gather our main results; for a graph $G$, by $|G|$ and $\|G\|$ we denote the cardinalities of the vertex and edge sets
of $G$, respectively.

\begin{theorem}[Constant factor approximation]\label{thm:main-apx}
Suppose $\F$ is a finite family of connected graphs and at least one member of $\F$ is planar and subcubic. 
Then there exists an algorithm that, given a graph $G$, runs in time $\Oh(\|G\|^3\log \|G\| \cdot |G|^3)$ and outputs a subset of edges $F\subseteq E(G)$ such that
$G-F$ is $\F$-immersion-free and the size of $F$ is at most $c_{\mathrm{apx}}$ times larger than the optimum size of a subset of edges with this property, 
for some constant $c_{\mathrm{apx}}$ depending on $\F$ only.
\end{theorem}

In Section~\ref{sec:conc} (Conclusions) we comment on how the constant-factor approximation can be generalized to work for $\F$ containing disconnected graphs as well, 
using the approach of Fomin et al.~\cite{FominLMS12,abs-1204-4230}.

\begin{theorem}[Linear kernelization and obstructions] \label{thm:main-ker}
Suppose $\F$ is a finite family of connected graphs and at least one member of $\F$ is planar and subcubic. 
Then there exists an algorithm that, given an instance $(G,k)$ of {\sc{$\F$-Immersion Deletion}}, runs in time $\Oh(\|G\|^4\log \|G\| \cdot |G|^3)$ and outputs an equivalent instance $(G',k)$ with $\|G'\|\leq c_{\mathrm{ker}}\cdot k$,
for some constant $c_{\mathrm{ker}}$ depending on $\F$ only. Moreover, there exists a constant $c_{\cal F}$ (non-constructively depending on ${\cal F}$) such that every graph
$H$ in ${\cal O}_{k}^{\rm im}$ has at most $c_{\cal F}\cdot k$  edges.
\end{theorem}

Thus, Theorems~\ref{thm:main-apx} and~\ref{thm:main-ker} mirror the approximation and kernelization results and the obstruction bounds of Fomin et al.~\cite{FominLMS12}.
However, this mirroring is not exact as we
even show that, in the immersion setting, a stronger
kernelization procedure can be designed. Namely, the size of the kernel given by Theorem~\ref{thm:main-ker} is {\em{linear}}, with only the multiplicative constant depending on the family $\F$, whereas
in the minor setting, the exponent of the polynomial bound on the kernel size provably must depend on $\F$ (under plausible complexity assumptions). This shows that the immersion and minor settings
behave quite differently and in fact stronger results can be obtained in the immersion setting. Observe that using Theorem~\ref{thm:main-ker} it is trivial to obtain a decision algorithm for
{\sc{$\F$-Immersion Deletion}} working in time $\Oh(c_{\mathrm{fpt}}^k + \|G\|^4\log \|G\| \cdot |G|^3)$ for some constant $c_{\mathrm{fpt}}$ depending on $\F$ only: one simply computes the kernel with a linear number of edges and
checks all the subsets of edges of size $k$.

%

\paragraph*{Our techniques.} 
Our approach to proving Theorems~\ref{thm:main-apx} and~\ref{thm:main-ker} roughly follows the general framework of protrusion replacement of Fomin et al.~\cite{FominLMS12}  (see also~\cite{BodlaenderFLPST09,BodlaenderFLPST09meta}).
We first define protrusions suited for the problem of our interest. In fact, our protrusions 
can be seen as the edge-analog of those introduced in~\cite{FominLMS12} (as in~\cite{Chatzidimitriou2015logopt}).
A protrusion for us is simply a vertex subset $X$ that induces an $\F$-immersion-free subgraph (which hence has constant tree-cut width, by Theorem~\ref{thm:woland}), and has a constant number of edges to the rest of the graph. When a large protrusion is localized, it can be replaced by a smaller gadget similarly as  in the work of Fomin et al.~\cite{FominLMS12}.
However, we need to design a new algorithm for searching for large protrusions, mostly in order to meet the condition that the exponent of the polynomial running time of the algorithm 
{\sl does not depend} on $\F$. For this, we employ the important cuts technique of Marx~\cite{Marx06} and the randomized contractions technique of Chitnis et al.~\cite{ChitnisCHPP12}. 
All of these yield an algorithm that exhaustively reduces all large protrusions.

Unfortunately, exhaustive protrusion replacement is still not sufficient for a linear kernel.
However, we prove that in the absence of large reducible protrusions, the only remaining obstacles are large groups of parallel edges between the same two endpoints (called {\em{thetas}}), 
and, more generally, large ``bouquets'' of constant-size graphs attached to the same pair of vertices.
Without these, the graph is already bounded linearly in terms of the optimum solution size.
The approximation algorithm can thus delete \emph{all} edges except for the copies included in bouquets and thetas, reducing the optimum solution size by a constant fraction of the deleted set.
It then exhaustively reduces protrusions in the remaining edges, and repeats the process until the graph is $\F$-immersion-free.

To obtain a linear kernel we need more work, as we do not know how to reduce bouquets and thetas directly.
Instead, we apply the following strategy based on the idea of {\em{amortization}}.
After reducing exhaustively all larger protrusions, we compute a constant-factor approximate solution $F_{\mathrm{apx}}$. 
Then we analyze the structure of the graph $G-F_{\mathrm{apx}}$, which has constant tree-cut width.
It appears that every bouquet and theta in $G$ can be reduced up to size bounded linearly in the number of solution edges $F_{\mathrm{apx}}$ that ``affect'' it.
After applying this reduction, we can still have large bouquets and thetas in the graph, but this happens only when they are affected by a large number of edges of $F_{\mathrm{apx}}$.
However, every edge of $F_{\mathrm{apx}}$ can affect only a constant number of bouquets and thetas and hence a simple amortization arguments shows that the total size of bouquets and theta is linear
in $|F_{\mathrm{apx}}|$, so also linear in terms of the optimum.

We remark that this part of the reasoning and in particular the amortization argument explained above, are fully new contributions of this work. 
These arguments deviate significantly from those needed by Fomin et al.~\cite{FominLMS12}, because they were aiming at a weaker goal of obtaining a polynomial kernel, instead of linear.
Also, we remark that, contrary to the work of Fomin et al.~\cite{FominLMS12}, all our algorithms are deterministic.

For the second part of Theorem~\ref{thm:main-ker}, we show that 
protrusions replacements can be done in a way that the resulting graph is an immersion of the original one. This implies 
that, in the equivalent instance  $(G',k)$ produced by our kernelization algorithm, the graph $G'$ is an immersion of $G$. Therefore
if $G$ is an immersion-obstruction of ${\cal G}^{\rm im}_{k-1,{\cal F}}$, then it should already have a linear, on $k$, number if edges (see Section~\ref{obstructions}).

\begin{wrapfigure}{r}{2.7cm}
  \begin{center}
\def\svgwidth{0.134\textwidth}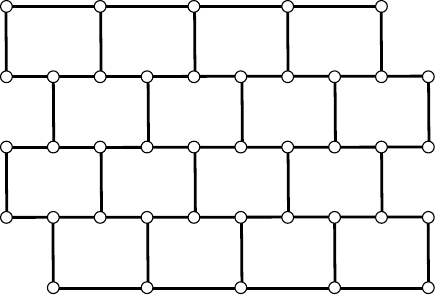
  \end{center}
 \caption{$W_{4,4}$.}\label{fig:wall}
\end{wrapfigure}
%

\paragraph*{Application: immersion-closed parameters.} 
Before we proceed to the proofs of Theorems~\ref{thm:main-apx} and~\ref{thm:main-ker}, we would like to highlight one particular meta-algorithmic application of our results which was our original motivation.
Suppose ${\bf p}$ is a graph parameter, that is, a function that maps graphs to nonnegative integers.
We shall say that ${\bf p}$ is {\em{closed under immersion}} if whenever a graph $H$ is an immersion of  another graph $G$, then ${\bf p}(H)\leq {\bf p}(G)$.
Furthermore, ${\bf p}$ is {\em{closed under disjoint union}} if ${\bf p}(G_1\uplus G_2)=\max({\bf p}(G_1),{\bf p}(G_2))$, for any two graphs $G_1$ and $G_2$; here, $\uplus$ denotes the disjoint union of two graphs.
Finally, ${\bf p}$ is {\em{large on walls}} if the set of integers $\{{\bf p}(W_{n,n})\}_{n\in \mathbb{N}}$ is infinite, where $W_{n,n}$ is the $n\times n$ wall, depicted in Figure~\ref{fig:wall}, for $n=4$.
The following proposition follows easily from Theorem~\ref{thm:woland} and the fact that the immersion order is a well-quasi-order.

\begin{proposition}\label{prop:par-obstr}
Let ${\bf p}$ be a graph parameter that is closed under immersion and under disjoint union and moreover is large on walls.
Then for every $r\in \mathbb{N}$ there exists a finite family of graphs $\F_{{\bf p},r}$ with the following properties: 
\begin{enumerate}[(a)]
\item every graph from $\F_{{\bf p},r}$ is connected;
\item $\F_{{\bf p},r}$ contains at least one planar subcubic graph; and
\item for every graph $G$, we have that ${\bf p}(G)\leq r$ if and only if $G$ is $\F_{{\bf p},r}$-immersion-free.
\end{enumerate}
\end{proposition}
\begin{proof}
Denote by $\G_{{\bf p},r}$ the class of all graphs $G$ for which ${\bf p}(G)\leq r$. 
Since ${\bf p}$ is immersion-closed, $\G_{{\bf p},r}$ is closed under taking immersions.
Since the immersion order is a well-quasi-order on graphs, we infer that there is a finite family $\F_{{\bf p},r}$ of graphs such that a graph $G$ belongs to $\G_{{\bf p},r}$ if and only if $G$ is $\F_{{\bf p},r}$-immersion-free.
Moreover, we can assume that $\F_{{\bf p},r}$ is minimal in the following sense: for each $H\in \F_{{\bf p},r}$ and each $H'$ that can be immersed in $H$ and is not isomorphic to $H$, we have that $H'\in \G_{{\bf p},r}$; 
equivalently, ${\bf p}(H')\leq r$.
We need to argue that every member of $\F_{{\bf p},r}$ is connected and that $\F_{{\bf p},r}$ contains a planar subcubic graph. 

For the first check, suppose that there is some disconnected graph $H$ in $\F_{{\bf p},r}$.
Then $H$ has two proper subgraphs $H_1$ and $H_2$ such that $H=H_1\uplus H_2$.
Since $H_1,H_2$ are subgraphs of $H$, they can, in particular, be immersed in $H$.
Both of them are strictly smaller than $H$, so we infer that ${\bf p}(H_1)\leq r$ and ${\bf p}(H_2)\leq r$.
As ${\bf p}$ is closed under disjoint union, we infer that 
${\bf p}(H)={\bf p}(H_1\uplus H_2)=\max({\bf p}(H_1),{\bf p}(H_2))\leq r.$
This is a contradiction to the fact that $H\notin \G_{{\bf p},r}$.

For the second check, since ${\bf p}$ is unbounded on walls, there is some integer $n$ such that ${\bf p}(W_{n,n})>r$.
Consequently, $W_{n,n}\notin \G_{{\bf p},r}$, so $W_{n,n}$ contains some graph $H$ from $\F_{{\bf p},r}$ as an immersion.
It can be easily seen that planar subcubic graphs are closed under taking immersions, so since $W_{n,n}$ is planar and subcubic, we infer that $H$ is also planar and subcubic.
\end{proof}

For a parameter ${\bf p}$ and a constant $r$, define the {\sc{${\bf p}$-at-most-$r$ Edge Deletion}} problem as follows: given a graph $G$ and an integer $k$, determine whether at most $k$ edges can be deleted from $G$
to obtain a graph with the value of ${\bf p}$ at most $r$. We also define the associated parameter ${\bf p}_{r}$ such that  $${\bf p}_{r}(G)=\min\{k\mid \exists S\subseteq E(G): |S|\leq k \wedge {\bf p}(G\setminus S)\leq r\}.$$
We also define ${\cal G}_{k,{\bf p}_{r}}=\{G\mid  {\bf p}_{r}(G)\leq k\}$
By combining Proposition~\ref{prop:par-obstr} with Theorems~\ref{thm:main-apx} and~\ref{thm:main-ker} we obtain the following corollary.

\begin{corollary}\label{cor:par-meta}
Let ${\bf p}$ be a graph parameter that is closed under immersion and under disjoint union and moreover is large on the class of walls\footnote{A graph parameter $\mathbf{p}$ 
is {\em large on a graph class $\mathcal{C}$} if $\{{\bf p}(G)\mid G\in{\cal C}\}$ is not a bounded set.}.
Then, for every constant $r$, the {\sc{${\bf p}$-at-most-$r$ Edge Deletion}} problem admits a constant-factor approximation and a linear kernel. Moreover, there is a constant $c_{r}$, depending (non-constructively) on $r$, such that for every $k$,
every graph
$H$ in ${\bf obs}_{\rm im}({\cal G}_{k,{\bf p}_{r}})$ has at most $c_{r}\cdot k$  edges.
\end{corollary}

Natural immersion-closed parameters that satisfy the prerequisites of Corollary~\ref{cor:par-meta} include cutwidth, carving width, tree-cut width, and edge ranking; 
see e.g.~\cite{SeymourT94,ThilikosSB05,Wollan15,Lam1998edge,IyerRV91onan} for more details on these parameters.
Corollary~\ref{cor:par-meta} mirrors the corollary given by Fomin et al.~\cite{FominLMS12} for the {\sc{Treewidth-$\eta$ Deletion}} problem, for which their results imply the existence of a constant-factor
approximation, a polynomial-kernel, a polynomial bound for the corresponding minor-obstruction set ${\bf obs}_{\rm mn}({\cal G}_{k,{\bf tw}_{\eta}})$, and a single-exponential FPT algorithm, for every constant $\eta$.

\paragraph*{Organization of the paper.} In Section~\ref{sec:prelims} we introduce notation, recall known definitions and facts, and prove some easy observations of general usage.
In Section~\ref{sec:nice} we provide several adjustments of the notions of tree-cut decompositions and tree-cut width.
In particular, we provide a simpler definition of tree-cut width that we use throughout the paper, and show that an optimum-width tree-cut decomposition may be assumed to have some additional, useful properties.
In Section~\ref{sec:protrusions} we discuss protrusions: finding them and replacing them.
Section~\ref{sec:approx} contains the proof of Theorem~\ref{thm:main-apx} (constant factor approximation), while Section~\ref{sec:kernel} contains the proof of Theorem~\ref{thm:main-ker} (linear kernelization). Section~\ref{obstructions} is dedicated to the linear bound on the size of obstructions.
We conclude with some finishing remarks in Section~\ref{sec:conc}.

\section{Preliminaries}\label{sec:prelims}

For a positive integer $p$, we denote $[p]=\{1,2,\ldots,p\}$.

\paragraph*{Graphs.}
In this work, all graphs are multigraphs without loops.
That is, a graph $G$ is a pair $(V(G),E(G))$, where $V(G)$ is the \emph{vertex set} and $E(G)$ is a multiset of \emph{edges}.
An \emph{edge} connects a pair of different vertices, called \emph{endpoints};
we write $uv \in E(G)$ for an edge with endpoints $u, v \in V(G)$.
Note that there might be several edges (called \emph{parallel edges}) between two vertices.
An edge is \emph{incident} to a vertex if that vertex is one of its two endpoints.

We write $|G|$ for $|V(G)|$ and $\|G\|$ for $|E(G)|$ (counting edges with multiplicities).
For a subset of vertices $X\subseteq V(G)$, $G[X]$ is the subgraph induced by $X$.
For a subset $X\subseteq V(G)$ of vertices, we write $G-X$ for the induced subgraph $G[V(G) \setminus X]$.
For a subset $F\subseteq E(G)$ of edges, we write $G-F$ for the graph obtained from $G$ by removing all edges of $G$, with $V(G-F)=V(G)$ and $E(G-F)=E(G)\setminus F$.

For two subsets $X,Y\subseteq V(G)$, not necessarily disjoint, $E_G(X,Y)$ denotes the set of all edges $xy\in E(G)$ for which $x\in X$ and $y\in Y$.
The \emph{boundary} of $X$ is $\delta_G(X) = E_G(X,V(G)\setminus X)$,
while the set of edges \emph{incident} to $X$ is $E_G(X, V(G))$.
For $v\in V(G)$, the \emph{degree} of $v$ is $\deg_G(v)=|\delta_G(\{v\})|$. 
We also define the set of neighbors of $v$: $N_G(v) = \{u \in V(G) : uv \in E(G)\}$.
By $N_G(X)$ we denote the open neighborhood of $X$, that is, the set of all vertices outside $X$ that have a neighbor in $X$.
We drop the subscript $G$ when it is clear from the context.
A graph is \emph{subcubic} if $\deg_G(v) \leq 3$ for every~$v\in V(G)$.

\paragraph*{Trees.}
A forest is a graph where every connected component is a tree.
For a forest $T$ and an edge $uv \in E(T)$, we denote by $T_{uv}$ and $T_{vu}$ the components of $T-uv$ containing $u$ and $v$, respectively.
Let $T$ be a rooted tree.
For every node $t\in V(T)$, we denote by $\pi(t)$ its unique parent on the tree~$T$. 
A node $t'$ is a {\em sibling} of $t$ if $t\neq t'$ and $t$ and $t'$ have the same parent.

\paragraph*{Tree-cut width.}
A \emph{near-partition} of a set $X$ is a family of (possibly empty) subsets $X_{1},\dots,X_{k}$ of $X$ such that 
$\bigcup_{i=1}^{k} X_{i} = X$ and $X_{i}\cap X_{j} =\emptyset$ for every $i\neq j$.

A {\em tree-cut decomposition} of a graph $G$ is a pair $\T=(T,{\cal X})$ such that $T$ is a forest and ${\cal X} = \{X_{t} :t \in V(T)\}$ is 
a near-partition of the vertices of $V(G)$.
Furthermore, we require that if $T_1,\dots, T_r$ are the connected components of $T$,
then $\bigcup_{t \in V(T_i)} X_t$ for $i\in[r]$ are exactly the vertex sets of connected components of $G$.
In other words, the forest $T$ has exactly one tree per each connected component of $G$, with this tree being a tree-cut decomposition of the connected component. 
We call the elements of $V(T)$ \emph{nodes} and the elements of $V(G)$ \emph{vertices} for clarity.
The set $X_{t}$ is called the {\em bag} of the decomposition corresponding to the node $t$, or just the {\em{bag at $t$}}.
By choosing a root in each tree of $T$, thus making $T$ into a rooted forest, we can talk about a \emph{rooted tree-cut decomposition}.

Let $G$ be a graph with a tree-cut decomposition $\T=(T, {\cal X} = \{X_{t} :t \in V(T)\})$.
For a subset $W\subseteq V(T)$, define $X_W$ as $\bigcup_{t\in W} X_t$.
For a subgraph $T'$ of $T$ we write $X_{T'}$ for $X_{V(T')}$.
For an edge $uv \in E(T)$ we write $X^T_{\ uv}$ for $X_{T_{uv}}$ to avoid multiple subscripts.
Notice that, since ${\cal X}$ is a near-partition, $\{X^T_{\ uv} , X^T_{\ vu}\}$ is a near-partition of the vertex set of a connected component of $G$.
We will call the decomposition \emph{connected} if for every edge $uv\in E(T)$, the graphs $G[X^T_{\ uv}]$ and $G[X^T_{\ vu}]$ are connected.



The {\em adhesion} of an edge $e = uv$ of $T$, denoted $\adh_{\T}(e)$, is defined as the set $E_G(X^{T}_{\ uv},X^{T}_{\ vu})$.
An adhesion is \emph{thin} if it has at most $2$ edges, and is \emph{bold} otherwise.
The {\em torso} at a node $t$ of $T$ is the graph $H^{\T}_{t}$ defined as follows. 
Let $T'$ be the connected component of $T$ that contains $t$, and let $T_{1},T_{2},\dots,T_{p}$ be the components of $T'-t$ (note there might be no such components if $t$ was an isolated node). 
Observe that $\{X_t , X_{T_1}, \dots, X_{T_p}\}$ is a near-partition of $X_{T'}$, whereas $X_{T'}$ induces a connected component of $G$.
Then the torso $H^{\T}_t$ is the graph obtained from $G[X_{T'}]$ by identifying the vertices of $X_{T_{i}}$ into a single vertex $z_{i}$, for each $i\in [p]$, and removing all the loops created in this manner.
Note that, thus, every edge between a vertex of $X_{T_{i}}$ and a vertex of $X_{T_{j}}$, for some $i\neq j$, becomes an edge between $z_i$ and $z_j$; similarly for edges between $X_{T_{i}}$ and $X_t$.
The vertices of $X_{t}$ are called the {\em core} vertices of the
torso, while the vertices $z_{i}$ are called the {\em peripheral}
vertices of the torso.
Finally, the \emph{3-center} of a node $t$ of $T$, denoted by $\overline{H^{\T}_t}$, is the graph obtained from the torso $H^{\T}_{t}$ 
by repeatedly suppressing peripheral vertices 
of degree at most two and deleting any resulting loops.
That is, any peripheral vertex of degree zero or one is deleted, while a peripheral vertex of degree two is replaced by an edge connecting its two neighbors;
if the two neighbors are equal, the resulting loop is deleted, potentially allowing further suppressions.
As Wollan~\cite{Wollan15} shows, any maximal sequence of suppressions leads to the same graph $\overline{H^{\T}_t}$.
We omit the subscripts and superscripts $\T$ when they are clear from the context.

The {\em width} of the decomposition $\T=(T,{\cal X})$, denoted $\width(\T)$, is 
\[
\max\{\max_{e\in E(T)} |\adh(e)| ,\max_{t\in
  V(T)} |V(\overline{H^\T_{t}})|\}.
\]
The {\em tree-cut width} of $G$, denoted by $\tctw(G)$, is the minimum width of a
tree-cut decomposition of~$G$.
\medskip

Ganian et al.~\cite{Ganian0S15} showed that bounded tree-cut width
implies bounded treewidth. Besides, Kim et al.~\cite{KimOPST15} showed that the dependency cannot be improved to subquadratic.

\begin{lemma}[see \cite{Ganian0S15}]\label{lem:twBound}
	For any graph $G$, $\tw(G)\leq 2 \tctw(G)^2 + 3 \tctw(G)$.
\end{lemma}

Finally, Kim et al.~\cite{KimOPST15} proposed a $2$-approximation FPT algorithm for computing the tree-cut width of a graph.

\begin{theorem}[see~\cite{KimOPST15}]\label{thm:apx}
There is an algorithm that, given a graph $G$ and an integer $r$, runs in time $2^{\Oh(r^2\log r)}\cdot |G|^2$ and either concludes that $\tctw(w)>r$, or returns a tree-cut decomposition of $G$ of width at most $2r$.
\end{theorem}

\paragraph*{Immersions.}
For two graphs $G$ and $H$ we say that \emph{$G$ contains $H$ as an immersion}, or $H$ is \emph{immersed} in $G$,
if there exist an injective mapping $\mu_V\colon V(H) \to V(G)$, and a mapping $\mu_E$ from edges of $H$ to paths in $G$ such that:
\begin{itemize}
	\item for any edge $uu'\in E(H)$, $\mu_E(uu')$ is a path in $G$ with endpoints $\mu_V(u)$ and $\mu_v(u')$; and
	\item for any pair of different edges $e,e' \in E(H)$, the paths $\mu_E(e)$ and $\mu_E(e')$ do not have common edges.
\end{itemize}

For a family of graphs $\F$, we say a graph $G$ is \emph{$\F$-immersion-free}, or \emph{$\F$-free} for short, if for every $H\in\F$, $G$ does not contain $H$ as an immersion.
We define the following parameterized problem:\medskip\medskip

\defparproblem{\textsc{\F-Immersion Deletion}}%
{A graph $G$ and a positive integer $k$.}{$k$}%
{Is there a set $F\subseteq E(G)$, such that $|F|\leq k$ and $G-F$ is $\F$-immersion-free?}
\medskip

By $\OPT_\F(G)$ we denote the minimum size of a set $F\subseteq E(G)$ such that $G-F$ is $\F$-free. If the family $\F$ is clear from the context, we omit the subscript.

\bigskip

The following result follows from the work of Wollan~\cite{Wollan15}; the improved bound is obtained using the polynomial Excluded Grid Minor Theorem by Chekuri and Chuzhoy~\cite{ChekuriC14,Chuzhoy15,Chuzhoy16}. 
In particular, note that Theorem~\ref{thm:woland} stated in the introduction follows from it.

\begin{theorem}\label{thm:Ffree_tctw}
	Let $\F$ be a family of graphs that contains at least one planar subcubic graph.
	Then if $G$ is an $\F$-free graph, then $\tctw(G)\leq a_\F$, where $a_\F=\Oh(\max_{H\in \F} (|H|+\|H\|)^{30})$ is a constant depending on $\F$ only.
\end{theorem}
\begin{proof}
	Recall that the $r\times r$ wall $W_{r,r}$ is a grid-like graph with maximum degree three (see Figure~\ref{fig:wall}).
	Theorem 17 of~\cite{Wollan15} states that if $G$ is a graph with tree-cut width at least $4r^{10} \cdot w(r)$, then $G$ admits an immersion of the $r\times r$ wall $W_{r,r}$.
	Here, $w(r)$ is the upper bound in the Excluded Grid Minor Theorem, that is, the maximum treewidth of a graph that excludes the $r\times r$ grid as a minor.
	The currently best upper bound for $w(r)$, given by Chuzhoy in~\cite{Chuzhoy16}, is $w(r) \leq \Oh(r^{19} \polylog(r))$.
	It is easy to see (see e.g.,~\cite{GiannopoulouKRT16pack}) that there is a constant $d$ such that 
	every planar subcubic graph $H$ with at most $r$ vertices and edges can be immersed into the $(dr)\times (dr)$-wall $W_{dr,dr}$.
	Hence, by the results above it follows that excluding such a graph $H$ as an immersion imposes an upper bound of $\Oh(r^{30})$ on the tree-cut width of a graph.
\end{proof}

The above theorem is a starting point for our algorithms.
In particular, it implies that we can test whether a graph is $\F$-free in linear time.

\begin{lemma}\label{lem:ffreeFpt}
	Let $\F$ be a family of graphs that contains at least one planar subcubic graph.		
	There is a linear-time algorithm that checks whether a given graph is $\F$-free.
\end{lemma}
\begin{proof}
	Let $a_\F$ be the constant given by Theorem~\ref{thm:Ffree_tctw} for the family $\F$, and let
	$G$ be the input graph. Using Bodlaender's algorithm~\cite{Bodlaender96} for computing treewidth, 
	for $k=2a_{\F}^{2}+3a_{\F}$ we either conclude that $\tw(G)>k$, or compute a tree decomposition of $G$ of width at most $k$.
	This takes time $f(k)\cdot |G|$ for some function $f$, hence linear time since $k$ is a constant.
	In the first case, when $\tw(G)>2a_{\F}^{2}+3a_{\F}$,
	we may directly conclude that $G$ is not $\F$-free, because from Lemma~\ref{lem:twBound} it follows that $\tctw(G)>a_{F}$, and then Theorem~\ref{thm:Ffree_tctw} implies that $G$ is not 
	$\F$-free. Thus, we may now assume that we have constructed a tree decomposition of $G$ of width at most $k=2a_{\F}^{2}+3a_{\F}$.
	
	Observe now that for any fixed graph $H$, the property of admitting $H$ as an immersion can be expressed in $\mathbf{MSO}_2$ (Monadic Second-Order logic on graphs with quantification over edge subsets). See, for example,~\cite{Giannopoulou2014effe}.
	Therefore, by applying Courcelle's Theorem~\cite{Courcelle90}, for every graph $H\in \F$ we may decide whether $G$ is $H$-free in time $g(k)\cdot \|G\|$ for some function $g$; that is, in linear time
	since $k$ is a constant. By verifying this for every graph $H\in \F$ we decide whether $G$ is $\F$-free.
\end{proof}
\bigskip

From now on, throughout the whole paper, we assume that $\F$ is a fixed family containing only connected graphs, of which at least one is planar and subcubic.
We define the constant $\MX=\max_{H\in \F} \|H\|$ and let $a_\F$ be the bound from Theorem~\ref{thm:Ffree_tctw}.

\section{Adjusting tree-cut decompositions}\label{sec:nice}
\subsection{Alternative definition of tree-cut width}

To simplify many arguments, we give a simpler definition of tree-cut width and show it to be equivalent.
Let $G$ be a graph with a  tree-cut decomposition $\T=(T,{\cal X} = \{X_t : t \in V(T) \})$. 
For every $t\in V(T)$, we define
\[
w_{\T}(t)=|X_{t}|+|\{t'\in N_{T}(t):\adh_\T(tt')~\text{is bold}\}|.
\]
We drop the subscript $\T$ when it is clear from the context.
We then set
\[
\width'(\T)=\max\{\max_{e\in E(T)}|\adh(e)|,\max_{t\in V(T)}w(t)\}
\] and define $\tctw'(G)$ as the minimum of $\width'(\T)$ over all tree-cut decompositions $\T$ of~$G$.

\begin{theorem}\label{thm:tcweqstcw'}
For every graph $G$ it holds that $\tctw(G)=\tctw'(G)$.
Moreover, given a tree-cut decomposition $\T=(T,{\cal X})$ of $G$, it always holds that $\width(\T)\leq \width'(\T)$, 
and a tree-cut decomposition $\T'$ such that $\width'(\T') \leq \width(\T)$ can be computed in time $\Oh(\|G\| \cdot |G|^2 \cdot \width(\T))$.
\end{theorem} 

\begin{proof} 
Let $G$ be a graph and $\T=(T,{\cal X} = \{X_t:t \in V(T)\}))$ be a tree-cut decomposition.
We first show that $\width(\T)\leq \width'(\T)$.
Consider any $t\in V(T)$ and let $t_{1},t_{2},\dots,t_{\ell}$ be the
neighbors of $t$ in $T$. 
For $i\in [\ell]$, denote by $z_{i}$ the peripheral vertex of the torso of $t$, $H^{\T}_{t}$, obtained after consolidating (i.e. identifying into one vertex) the set~$X^T_{\ t_{i} t}$. 

We claim that $|V(\overline{H^{\T}_{t}})|\leq w(t)$.
Recall that $\overline{H^{\T}_{t}}$ is the 3-center at $t$: its vertices are core vertices $X_t$ and peripheral vertices $z_1,\dots,z_\ell$ that were not suppressed.
Since $\deg_{H^{\T}_{t}}(z_{i})=|\adh(t t_{i})|$ for $i\in [\ell]$, 
the vertex $z_i$ might not suppressed (and thus, belong to the 3-center at $t$) only when $|\adh(t t_i)|\geq 3$, which implies
\begin{equation}\label{eq:eq1}
|V(\overline{H^{\T}_{t}})| \leq |X_{t}|+|\{t'\in N_{T}(t):\adh(tt')~\text{is bold}\}| = w(t)
\end{equation}
This holds for every $t\in V(T)$, hence $\width(\T)\leq \width'(\T)$.

\bigskip

In particular, $\tctw(G)\leq \tctw'(G)$. We now proceed to showing that $\tctw(G)=\tctw'(G)$. 
Note that without loss of generality we may assume that $G$ is connected, as we may consider  each connected component separately.
Hence, all the tree-cut decompositions considered in the sequel will consist of just one tree.

Let us choose a tree-cut decomposition $\T=(T,{\cal X})$ of $G$ as follows: 
$\T$ has the optimum width (i.e. $\width(\T)=\tctw(G)$) and, among such optimum decompositions, $\sum_{e\in E(T)}|\adh(e)|$ is minimum possible.
We will now prove that $|V(\overline{H^{\T}_{t}})|= w(t)$ for every $t\in V(T)$, 
which implies that $\width'(\T) = \width(\T)=\tctw(G)$ by definition, concluding the claim.

Towards a contradiction, let us assume that there exists $t\in V(T)$ with $|V(\overline{H^{\T}_{t}})|\neq w(t)$, and thus by \eqref{eq:eq1}, $|V(\overline{H^{\T}_{t}})|< w(t)$.
We denote by $t_{1},t_{2},\dots,t_{\ell}$ the
neighbors of $t$ in $T$ and define $z_i$ for $i\in [\ell]$ as above.

Without loss of generality, let $(z_{1},z_{2},\dots,z_{k})$ be a maximal sequence of vertices whose suppression leads to $\overline{H^{\T}_{t}}$; see Figure~\ref{fig:widthprime}. 
Observe that, since the inequality $|V(\overline{H^{\T}_{t}})|< w(t)$ is strict, at least one of the suppressed vertices originally has degree 3 or more in $H^{\T}_{t}$.
Let $z_{p}$ be the first such vertex, that is, $\deg_{H^{\T}_{t}}(z_{p})\geq 3$ and  $\deg_{H^{\T}_{t}}(z_{i}) \leq 2$ for $i<p$.
Note that $p\neq 1$ as, by definition, the first vertex to get suppressed has degree at most 2 in $H^{\T}_{t}$. 
Let $C_{1},\dots,C_{p'}$ be the connected components of the graph induced by the vertices $z_{1},\dots,z_{p-1}$ in $H^{\T}_{t}$. 
Since their degrees are at most two each component $C_i$ is either an induced cycle or an induced path in $H^{\T}_t$.
Let ${\cal C}=\{C_{i}:N_{H^{\T}_{t}}(z_{p})\cap V(C_{i})\neq \emptyset\}$, that is, let ${\cal C}$ be the subset of the graphs 
$C_{1},\dots, C_{p'}$ in which $z_{p}$ has a neighbor. 

First notice that the set ${\cal C}$ is not empty.
Indeed, if $z_{p}$ would not have any neighbor in at least one the components $C_{1},\dots,C_{p'}$ 
then it still would have degree at least 3 in $H^{\T}_{t}$ after suppressing the vertices of these graphs, $z_1,\dots,z_{p-1}$.

Second, for $C\in {\cal C}$, observe that if $z$ is a neighbor of $z_p$ in $C$, then $z$ has degree at most 2 in $H^{\T}_t$, including its neighbor $z_p$, and hence $z$ has degree at most 1 in $C$. 
Therefore, every $C\in {\cal C}$ is an induced path in $H^{\T}_t$, whose endpoints we henceforth denote by $z_C$ and $z_C'$.
We have $N_{H^{\T}_t}(z_{p})\cap V(C)\subseteq \{z_{C},z_{C}'\}$,
and since by definition $N_{H^{\T}_t}(z_{p})\cap V(C) \neq \emptyset$, without loss of generality we will assume that
always $z_{C}z_{p}\in E(H^{\T}_{t})$.

\begin{claim}
	There exists $C\in {\cal C}$ such that $N_{H^{\T}_{t}}(V(C)) \subseteq \{z_{p}\}$.
\end{claim}
\begin{proof}
	We prove the claim by contradiction. That is, suppose that for every $C\in {\cal C}$, $C$ has a neighbor in $V(H^{\T}_{t})\setminus \{z_p\}$.
	This must be a neighbor outside $\{z_i: i\in [p-1]\}$ (since $C$ is a connected component of the subgraph induced by these vertices), 
	thus for every $C\in {\cal C}$, there is a vertex $z\in C$ such that $z$ has a neighbor $z'$ in $V(H^{\T}_{t})\setminus \{z_{i}: i\in [p]\}$.
	
	Let $C\in {\cal C}$.
	Since the internal vertices of the path $C$ have degree exactly 2 both in $C$ and in $H^{\T}_t$, they have no neighbors outside $C$. Consider now two cases depending on $|V(C)|$.
        \begin{itemize}
        \item If $|V(C)|=1$, then $z_C=z_C'$ has an edge to a neighbor $z'$ in $V(H^{\T}_{t})\setminus \{z_{i}: i\in [p]\}$ and an edge to $z_p$; since $z_C=z_C'$ has degree 2 in $H^{\T}_t$, 
	it has no other incident edges in $H^{\T}_t$.
      \item If $|V(C)|\geq 2$, then $z_C$ has an edge to a neighbor in $C$ and an edge to $z_p$, and again, no other incident edges.
	Thus it must be that $z_C'$ has a neighbor $z'$ in $V(H^{\T}_{t})\setminus \{z_{i}: i\in [p]\}$.
	Then $z_C'$ has an edge to $z'$, to a neighbor in $C$, and no other incident edges. 
        \end{itemize}
	
	We conclude that in both cases, for every $C\in{\cal C}$, we have $|E(C,z_p)|=1$;
	moreover, after suppressing the vertices $z_1,\dots,z_{p-1}$, including those of $C$, $z_p$ has an edge to $V(H^{\T}_{t})\setminus \{z_{i}: i\in [p]\}$, a different one for every $C \in {\cal C}$.
	Therefore, the degree of $z_p$ in $H^{\T}_t$ does not drop after suppressing $z_1,\dots,z_{p-1}$.
	However, initially $\deg_{H^{\T}_t}(z_p)\geq 3$ by choice of $z_p$, and after suppressing $z_1,\dots,z_{p-1}$, the vertex $z_p$ must have degree at most 2 to be itself suppressed, a contradiction.	
\cqed\end{proof}

	\begin{figure}[H]
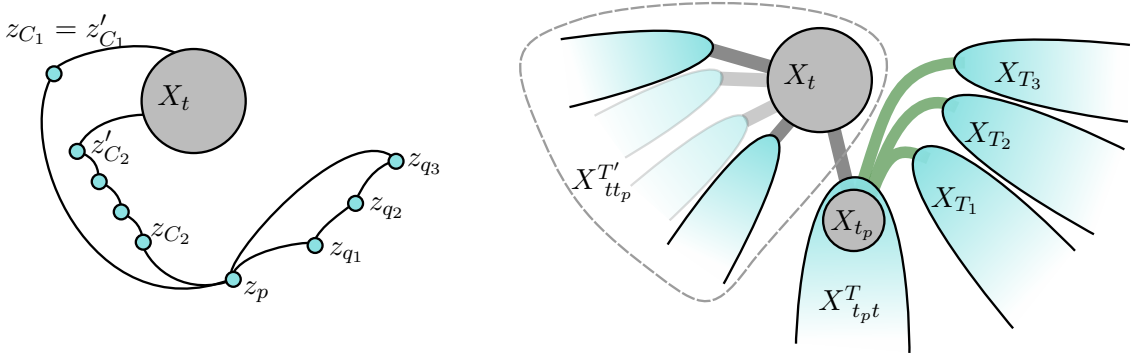

		\centering
		\svg{0.95\textwidth}{widthprime}
		\caption{
			On the left, the torso $H^{\T}_t$, including the peripheral vertices $z_1, \dots, z_p$ suppressed in $\overline{H^{\T}_{t}}$.
			On the right, the modified decomposition, with subtrees $T_1, \dots, T_r$ attached to $t_p$ instead of $t$.
			This improves the decomposition by making $\adh(t t_p)$ strictly smaller (because of $z_p z_{q_1}$, $z_p z_{q_3}$).}
		\label{fig:widthprime}
	\end{figure}

Let now $C \in {\cal C}$ be such that $N_{H^{\T}_{t}}(C)\subseteq \{z_p\}$.
Let $z_{q_{1}},z_{q_{2}},\dots,z_{q_{r}}$ be the vertices of the path $C$. We construct a new tree-cut decomposition $\T'=(T',{\cal X})$ of $G$ by removing the edges
$tt_{q_{i}}$ from $T$ and adding the edges $t_{p}t_{q_{i}}$, for all $i\in [r]$.
That is, $$V(T')=V(T) \qquad \textrm{and} \qquad E(T')=(E(T)\setminus \{t t_{q_i} :i\in [r]\})\cup\{t_p t_{q_i}:i\in [r]\}.$$
Note that the bags in $\T'$ are exactly the same as in $\T$.
We will show that $\T'=(T',{\cal X})$ has the minimum possible width, and that
$\sum_{e\in E(T')}|\adh_{T'}(e)|<\sum_{e\in E(T)}|\adh_{T}(e)|$, a contradiction to the choice of $\T$.
Notice that $T_{t_{q_i} t} = T'_{t_{q_i} t_p}$ for each $i \in [r]$; we will denote this subtree of $T$ and $T'$ by $T_i$ from now on.
The bags assigned to this subtree do not change either, thus for every $i\in [r]$ we have
\begin{equation}\label{eq:eq2}
\adh_{\T'}(t_{q_i} t_p)=\adh_{\T}(t_{q_i} t).
\end{equation}
Similarly, notice that for every edge $e\in (E(T)\cap E(T'))\setminus \{t t_{p}\}$ we have
\begin{equation}\label{eq:eq3}
\adh_{\T'}(e)=\adh_{\T}(e).
\end{equation}
Finally, let us consider the edge $t t_{p}$. By construction 
$V(T_{t t_p}') =  V(T_{t t_p})\setminus \bigcup_{i\in [r]} V(T_i)$.
In other words, $V(G)$ is partitioned into $X^T_{\ t_p t}$, $X^{T'}_{\ t t_p}$, and $\bigcup_{i\in [r]} X_{T_i}$.
Observe that $\adh_{\T'}(tt_p)$ is obtained from $\adh_{\T}(tt_p)$ by deleting $E_G(X_{T_i}, X^T_{\ t_p t})$ and adding $E_G(X_{T_i}, X^{T'}_{\ t t_p})$, for all $i\in [r]$.
\begin{claim}
	For every $i\in [r]$, $E_G(X_{T_i},X^{T'}_{\ t t_p})=\emptyset$.
	Moreover, there exists $i\in [r]$, $E_G(X_{T_i},X^{T}_{\ t_p t})\neq\emptyset$.
\end{claim}
\begin{proof}
	By the choice of $C$, the only neighbor of $V(C)=\{z_{q_1},\dots,z_{q_r}\}$ in $H^{\T}_t$ is $z_p$.
	That is, for each $i\in[r]$, the vertex $z_{q_i}$ has neighbors only in $z_{q_1},\dots,z_{q_r}$ and $z_p$ in $H^{\T}_t$.
	Since $H^{\T}_t$ is constructed from $G$ by consolidating $X_{T_i}$ into $z_{q_i}$ (for $i\in [r]$) and $X^T_{\ t_p t}$ into $z_p$ (among others),
	this means that the only edges in $G$ leaving $X_{T_i}$ go to $X_{T_1} \cup \dots \cup X_{T_r} \cup X^T_{\ t_p t}$.
	Since $X^{T'}_{\ t t_p}$ is obtained from $X^T_{\ t t_p} = V(G) \setminus X^T_{\ t_p t}$ by removing $X_{T_1} \cup \dots \cup X_{T_r}$, the first claim follows.
	
	Similarly, by the choice of $C$, $z_p$ has an edge to $V(C)=\{z_{q_1},\dots,z_{q_r}\}$ in $H^{\T}_t$.
	This means $G$ has an edge between $X^T_{\ t_p t}$ and $X_{T_i}$ for  some $i\in[r]$.
\cqed\end{proof}

\noindent
It follows that $\adh_{\T'}(tt_{p}) \subsetneq \adh_{\T}(tt_{p})$.
Together with~\eqref{eq:eq2}, \eqref{eq:eq3}, this implies
\begin{eqnarray}
\max_{e\in E(T)}|\adh_{\T'}(e)| & \leq & \max_{e\in E(T)}|\adh_{\T}(e)|\label{eq:eq5}\\ 
\sum_{e\in E(T')}|\adh_{\T'}(e)| & < & \sum_{e\in E(T)}|\adh_{\T}(e)|.\label{eq:eq6}
\end{eqnarray}

To prove our claim it remains to show that $\max_{t\in V(T)}|V(\overline{H^{\T}_{t}})|\leq \max_{t\in V(T)}|V(\overline{H^{\T'}_{t}})|$.
Recall first that $V(T)=V(T')$.
For every vertex $s\in V(T)\setminus \{t,t_{p}\}$ the torso at $s$ in decomposition $\T$ is the same as the torso at $s$ in decomposition $\T'$. This implies that $\overline{H^{\T}_{s}}=\overline{H^{\T'}_{s}}$.

Consider now the torso at $t$ in decomposition $\T'$, i.e. $H^{\T'}_t$. From the way $T'$ was constructed, we have $N_{T'}(t)\subseteq N_{T}(t)$. 
Moreover, recall that for every vertex $s\in N_{T'}(t)$, $\adh_{T'}(st)\subseteq \adh_{T}(st)$. Therefore $H^{\T'}_{t}$ is a subgraph of $H^{\T}_{t}$, and thus
$|V(\overline{H^{\T'}_{t}})|$ cannot be larger than $|V(\overline{H^{\T}_{t}})|$.

Consider finally the torso at the node $t_{p}$ in decomposition $\T'$, i.e. $H^{\T'}_{t_{p}}$. Notice that $V(H^{\T'}_{t_{p}})=V(H^{\T}_{t_p})\cup \{z_{q_{i}}:i\in [r]\}$. Recall, however, that 
$\adh_{\T'}(t_p t_{q_i})=\adh_{\T}(t t_{q_i})$ and therefore, $\deg_{H^{\T'}_{t_p}}(z_{q_i}) = \deg_{H^{\T}_t}(z_{q_i})\leq 2$. 
This implies that the vertices $z_{q_{i}}$, for all $i\in [r]$, get suppressed in $\overline{H^{\T'}_{t_{p}}}$ (more precisely, we can start the procedure of obtaining $\overline{H^{\T'}_{t_{p}}}$ by suppressing them).
For other vertices $s$ in $N_{T'}(t_p)$ we have $\adh_{\T'}(st_p)\subseteq \adh_{\T}(st_p)$.
Hence $|V(\overline{H^{\T'}_{t_p}})|$ cannot be larger than $|V(\overline{H^{\T}_{t_p}})|$.

In any case we obtain that $|V(\overline{H^{\T'}_{t}})|\leq |V(\overline{H^{\T}_{t}})|$ for every $t\in V(T)$.
Together with~\eqref{eq:eq5}, we obtain that $\width(\T')\leq \width(\T)$, that is,
$\T'$ has the minimum possible width (because we assumed $\T$ has).
But then~\eqref{eq:eq6} contradicts our choice of $\T$.
Hence, $|V(\overline{H^{\T}_{t}})|= w(t)$, for every $t\in V(T)$, that is, 
$\width(\T) = \width'(\T)$, which concludes the proof that $\tctw(G)=\tctw'(G)$.

For the algorithmic statement, note that in the proof, either we show that $\width(\T) = \width'(\T)$, 
or we construct a decomposition $\T'$ such that $\width(\T') \leq \width(\T)$ and the sum of all adhesion sizes in $\T'$ is strictly smaller than in $\T$ (Equation~\eqref{eq:eq6}).
Since the construction can be easily performed in time $\Oh(\|G\|\cdot |G|)$ by computing all torsos and 3-centers, 
and since the initial sum of all adhesion sizes can be at most $\Oh(|G|\cdot \width(\T))$, 
the construction can be performed repeatedly until $\width(\T) = \width'(\T)$, concluding the algorithm.
\end{proof}

\subsection{Making the decomposition connected}

The goal of this section is to prove that one can make a tree-cut decomposition connected without increasing its width (more precisely, $\width'$) by much.
For this, we will temporarily need a slight variation on $\tctw'$.
Let $G$ be a graph and $\T=(T,{\cal X})$ be a tree-cut decomposition of $G$.
For every $t\in V(T)$, recall that $\delta_T(t)$ denotes the set of edges of $T$ incident to $t$. We define 
\begin{align*}
z(t)&=|X_{t}|+\sum_{\substack{e\in \delta(t)\\ |\adh(e)|\geq 3}}|\adh(e)|\quad,\ \text{and}\\
\width''(\T)&=\max_{t\in V(T)}z(t).
\end{align*}

\begin{lemma}\label{lem:boundondoubleprimewdth}
For every graph $G$ and every tree-cut decomposition $\T$ of $G$, 
$$\width'(\T)-1\ \leq\ \width''(\T)\ \leq\ (\width'(\T))^{2}.$$
\end{lemma}

\begin{proof}
Let $G$ be a graph and $\T=(T,{\cal X}=\{X_t : t \in V(T)\})$ be a tree-cut decomposition of $G$.
We first prove that $\width'(\T)\leq \width''(\T)+1$, which is equivalent to the left inequality.
Recall that $\width'(\T)=\max\{\max_{e\in E(T)}|\adh(e)|,\max_{t\in V(T)}w(t)\}$.
Clearly for each $t\in V(T)$, from definitions of $w(t)$ and $z(t)$ we have
$$ w(t) = |X_t| + |\{e \in \delta_T(t) :~|\adh(e)|\geq 3\}|\ \leq\ |X_t| + \sum_{\substack{e\in \delta(t)\\|\adh(e)|\geq 3}}|\adh(e)| = z(t).$$
In particular, $\max_{t \in V(T)} w(t)\leq \max_{t \in V(T)} z(t)$.

Let $e^*\in E(T)$ be an edge of $T$ that maximizes $|\adh(e^*)|$.
If $|\adh(e^*)| \leq 2$, then trivially $|\adh(e^*)| \leq |X_t|+1 \leq z(t)+1$ for some $t \in V(T)$.
Otherwise, if $|\adh(e^*)| \geq 3$, let $t$ be an endpoint of $e^*$;
then $z(t) \geq |X_t| + |\adh(e^*)|$, so in particular $|\adh(e^*)| \leq z(t)$.
Thus in any case, we conclude that
$$\width'(\T) =\max\{|\adh(e^*)|, \max_{t \in V(T)} w(t)\} \leq \max_{t \in V(T)} z(t) +1 = \width''(\T) +1.$$

We now prove that $\width''(\T)\leq (\width'(\T))^{2}$.
Notice  that, by definition, 
$$\max_{e\in E(T)}|\adh(e)|\leq \width'(\T).$$
Moreover, for every $t\in V(T)$,  
$$|X_{t}| + |\{e \in \delta_T(t) : |\adh(e)|\geq 3\}|\ =\ w(t)\ \leq\ \width'(\T).$$
Therefore, 
\begin{align*}
z(t) &= |X_t|+\sum_{\substack{e\in \delta(t)\\ |\adh(e)|\geq 3}} |\adh(e)|\\
     &\leq |X_t|+\sum_{\substack{e\in \delta(t)\\ |\adh(e)|\geq 3}} \width'(\T)\\
     &= |X_t| + |\{e\in \delta(t) :|\adh(e)|\geq 3\}| \cdot \width'(\T)\\
     &\leq\ (\width'(\T))^2.
\end{align*}
Thus, we conclude that $\width''(\T) = \max_{t\in V(T)} z(t) \leq  (\width'(\T))^2$.
\end{proof}

Recall that a tree-cut decomposition $\T=(T,{\cal X}=\{X_t : t \in V(T)\})$ of $G$ is \emph{connected} if for every edge $uv\in E(T)$, the graphs $G[X^T_{\ uv}]$ and $G[X^T_{\ vu}]$ are connected.
We now show that we may always find such
a tree-cut decomposition.

\begin{lemma}\label{lem:conctvt}
Given a tree-cut decomposition of a graph $G$ with $\width'$ at most $k$, 
a connected tree-cut decomposition of $G$ with $\width'$ at most $k^2+1$ can be constructed in time $\Oh(\|G\| \cdot |G|^2 \cdot k^2)$.
\end{lemma}
\begin{proof}
If a graph is disconnected, we may consider its connected components separately, find a connected tree-cut decomposition for each component and conclude the claim by taking the disjoint union of the decompositions.
We will thus henceforth assume that $G$ is a connected graph.

Suppose $G$ has a tree-cut decomposition of $\width'$ at most $k$.
Then by Lemma~\ref{lem:boundondoubleprimewdth}, the same decomposition has $\width''$ at most $k^2$.
Let $\T=(T,{\cal X}=\{X_t : t \in V(T)\})$ be a tree-cut decomposition of $G$ such that $\width''(\T)\leq k^2$ and, subject to that, $\sum_{e\in E(T)}|\adh_\T(e)|^{2}$ is minimum possible.

We claim that $\T$ is a connected decomposition.
Towards a contradiction, assume that $G[X^T_{\ uv}]$ is not connected, for some $uv \in E(T)$.
Let $G_{1},G_{2},\dots, G_{r}$ be its connected components. 
Let $\T'=(T',{\cal X}')$ be the tree-cut decomposition of $G$ obtained from $(T,{\cal X})$ in the following way. Let $T_{1},T_{2},\dots, T_{r}$ be $r$ distinct copies of
$T_{uv}$ where, for every $i\in [r]$,
\begin{eqnarray*}
V(T_{i}) & = & \{z_{i}:z\in V(T_{uv})\}\\ 
E(T_{i}) & = & \{f_{i}: f\in E(T_{uv})\}.
\end{eqnarray*} 
To obtain $T'$, we remove $T_{uv}$ from $T$, add the trees $T_{1},T_{2},\dots,T_{r}$ instead, and for each $i=1,2,\ldots,r$ we add a new edge 
$e_{i}$ between $v$ and $u_{i}\in V(T_{i})$.
Therefore,
\begin{eqnarray}
V(T') & = & V(T_{vu})\cup \bigcup_{i\in [r]}V(T_{i})\\
E(T') & = & E(T_{vu})\cup \bigcup_{i\in [r]}E(T_{i})\cup \bigcup_{i\in [r]} \{e_{i}\}.\label{eq:eqedges} 
\end{eqnarray}
Notice then that $T_{i}=T'_{u_{i}v}$, for each $i\in [r]$.
We define ${\cal X}'$ in the following way. 
\begin{equation}\label{eq:eq22}
X_{s}'=\begin{cases}
X_{s} & \text{if } s\in V(T_{vu})\\
X_{z}\cap V(G_{i}) & \text{if } s=z_{i} \text{ for some } z\in V(T_{uv})\text{ and }i\in [r].
\end{cases}
\end{equation}
%
%
%
We will obtain a contradiction by proving that 
$$\width''(\T')\leq \width''(\T)\qquad \textrm{and}\qquad \sum_{e\in E(T')}|\adh_{\T'}(e)|^{2}<\sum_{e\in E(T)}|\adh_{\T}(e)|^{2}.$$
However, towards our goal, we have
to first show how adhesions in $\T'$ correspond to adhesions in $\T$.


Notice that for every $f\in E(T)\cap E(T')$,
\begin{equation}\label{eq:eq23}
\adh_{\T'}(f)=\adh_{\T}(f).
\end{equation}
The remaining edges of $T'$ are of the form $p_{j}q_{j}\in E(T_{j})$ or $e_{j}$, for some $j\in [r]$. In the latter case, let us write $p_{j}=u_{j}$ and $q_{j}=v$.
We claim that $\{ \adh_{\T'}(p_j q_j)   : j \in [r] \}$ is a near-partition of $\adh_{\T}(pq)$ (here, if $p_jq_j=e_j=u_jv$, then $p=u$ and $q=v$).
Indeed, $\adh_{\T'}(p_j q_j)$ is by construction equal to the set of edges between $X^T_{\ pq}\cap V(G_{j})$ and $X^T_{\ qp}\cup \bigcup_{i\neq j}V(G_{i})$.
Since, $E(V(G_{j}), V(G_{i}))=\emptyset$, for $i\neq j$, it follows that
$$\adh_{\T'}(p_{j}q_j)=E_G(X^T_{\ pq}\cap V(G_{j}), X^T_{\ qp}).$$
Since $X^T_{\ qp}, X^T_{\ pq}$ is a near-partition of $V(G)$ and $\{X^T_{\ pq} \cap V(G_{i}) : i \in [r]\}$ is a near-partition of $X^T_{\ pq}$, we infer that
 $\{ \adh_{\T'}(p_j q_j)   : j \in [r] \}$ is a near-partition of $E_G(X^T_{\ pq}, X^T_{\ qp})=\adh_{\T}(pq)$, as claimed.
 Therefore, for all edges $f \in E(T_{vu})$, as well as for $f=uv$ (in which case $f_i=e_i$), we have
\begin{equation}\label{eq:eq29}
\sum_{i\in [r]}|\adh_{\T'}(f_{i})| = |\adh_{\T}(f)|.
\end{equation}

We are now able to prove that $\width''(\T')\leq \width''(\T)$.
That is, we want to show that 
 $\max_{t \in V(T')} z_{\T'}(t)\leq \max_{t \in V(T)} z_{\T}(t)$. 
 Here $z_{\T}(\cdot)$ and $z_{\T'}(\cdot)$ are the $z(\cdot)$-functions as in the definition of $\tctw''$, applied respectively in decompositions $\T$ and $\T'$.

Let first $t\in V(T_{vu})\setminus \{v\}$. From~\eqref{eq:eq22} we obtain that $|X_{t}'|=|X_{t}|$ and, from~\eqref{eq:eq23} 
we obtain that for every edge $y\in \delta_T(t)$, $|\adh_{\T'}(y)|=|\adh_{\T}(y)|$.
Thus, $z_{\T'}(t)=z_{\T}(t)$.

Let now $t_{i}\in V(T_{i})$, for some $i\in [r]$. From~\eqref{eq:eq22} we obtain that $|X_{t_{i}}'|\leq |X_{t}|$. 
Moreover, for every edge $f_i$ incident to $t_i$, from~\eqref{eq:eq29} we obtain that $|\adh_{\T'}(f_i)|\leq |\adh_\T(f)|$. Thus, $z_{\T'}(t_{i})\leq z_{\T}(t)$.

Finally, let $t=v$. From~\eqref{eq:eq22}, we obtain that $|X_{v}'|=|X_{v}|$.
Observe that $\delta_{T'}(v) = E_1 \uplus E_2$, where
$E_{1}=\delta_T(v)\setminus \{uv\}$ and $E_{2}=\{e_{i}\mid i\in [r]\}$.
Then from Equation~\eqref{eq:eq23}, for every edge $y\in E_{1}$ we have $|\adh_{\T'}(y)|=|\adh_{\T}(y)|$, and from Equation~\eqref{eq:eq29}, we also have $\sum_{i\in[r]}|\adh_{\T'}(e_{i})| = |\adh_{\T}(e)|$.
From this it follows that $z_{\T'}(v) \leq z_{\T}(v)$.

\smallskip
Thus for each $t'\in V(T')$, we have $z_{\T'}(t')\leq \max_{t\in V(T)} z_{\T}(t) = \width''(\T)$,
and hence  $\width''(\T')\leq \width''(\T)$.
Furthermore, from~\eqref{eq:eq23} and~\eqref{eq:eq29}, we have 
$$\sum_{e \in E(T')} |\adh_{\T'}(e)| = \sum_{e \in E(T)} |\adh_{\T}(e)|$$

Notice now that from~\eqref{eq:eq29}, for each $f \in E(T_{vu})\cup \{uv\}$  we have
\begin{equation}\label{eq:eq34}
\sum_{i\in [r]}|\adh_{\T'}(f_{i})|^{2} \leq \bigl(\sum_{i\in [r]}|\adh_{\T'}(f_{i})|\bigr)^{2} = |\adh_{\T}(f)|^{2}.
\end{equation}
Here, for $f=uv$ we consider $f_i=e_i$.

Since $X^T_{\ uv},X^T_{\ vu}$ is a near-partition of $V(G)$ and $V(G_1),\dots,V(G_r)$ is a partition of $X^T_{\ uv}$, we have that $X^T_{\ uv},V(G_1),\dots,V(G_r)$ is a near-partition of all of $V(G)$.
Furthermore, since $G$ is connected and $E_G(V(G_i),V(G_j))=\emptyset$ for $i\neq j$, it must be that $E_G(V(G_i),X^T_{\ uv})$ is non-empty for each $i\in [r]$. 
This means $\adh_{\T'}(e_{i})$ is non-empty, and since $r\geq 2$, we infer that the inequality in~\eqref{eq:eq34} is strict for $f=uv$.
%
We conclude that 
\begin{equation}\label{eq:connectivityPotential}
\sum_{e\in E(T')}|\adh_{\T'}(e)|^{2}<\sum_{e\in E(T)}|\adh_{\T}(e)|^{2},
\end{equation}
a contradiction to the choice of $\T$.

This concludes the proof that $G$ has a connected tree-cut decomposition of $\width''$ at most $k^2$ and hence, by Lemma~\ref{lem:boundondoubleprimewdth}, of $\width'$ at most $k^2+1$.
Note that in the proof, we either proved that $\T$ is already connected, or constructed a decomposition $\T'$ with $\width''(\T') \leq \width''(\T)$ 
and with a strictly smaller sum of squares of adhesion sizes (Equation~\eqref{eq:connectivityPotential}).
Since the construction can be easily performed in $\Oh(\|G\|\cdot |G|)$ time, and since the  sum of squares of adhesion sizes is initially bounded by $\Oh(|G|\cdot k^2)$ 
(as each adhesion has size bounded by $\width'$ of the decomposition, which is at most $k$), the construction can be performed repeatedly until $\T$ is connected, concluding the algorithm.
\end{proof}

\subsection{Neat tree-cut decompositions} 
Recall that a tree-cut decomposition can be rooted by selecting a root in every its tree, which naturally imposes child-parent relation on the nodes, as well as the sibling relation.
As already mentioned, the parent of a node $t$ is denoted by $\pi(t)$.
We now define additional properties of rooted tree-cut decompositions, and show that these properties can be achieved by simple modifications of the decomposition.
This will help us in the next sections, where we will handle tree-cut decompositions combinatorially and algorithmically.
The main notion that we will be interested in is called {\em{neatness}}.

\begin{definition}
Let $G$ be a graph and $\T=(T,{\cal X}=\{X_t:t\in V(T)\})$ be a rooted tree-cut decomposition of $G$. 
We say that $\T$ is {\em neat} if it is connected and, furthermore, for every non-root node $t\in V(T)$ such that $\adh(t\pi(t))$ is thin, and every sibling $t'$ of $t$, 
there are no edges between $X^T_{\ t\pi(t)}$ and $X^T_{\ t'\pi(t)}$ in $G$.
\end{definition}

The second condition was used by Ganian et al.~\cite{Ganian0S15} under the name {\em{niceness}}.
We now show that every connected tree-cut decomposition can be made neat without increasing its width.
The proof of this result follows closely the lines of the proof of~\cite[Lemma~1]{Ganian0S15}.

\begin{theorem}\label{thm:neat}
Given a connected tree-cut decomposition of a graph $G$ with $\width'$ at most $k$, a neat tree-cut decomposition of $G$ with $\width'$ at most $k$ can be computed in time $\Oh(|G|^3)$.
\end{theorem}

\begin{proof}
Similarly as before, if a graph is disconnected, we may consider its connected components separately, find a neat tree-cut decomposition for each component and conclude the claim by joining the decompositions into one forest. We will thus henceforth assume that $G$ is a connected graph.

Let $\T=(T,{\cal X})$ be a connected tree-cut decomposition of $G$ with $\width'(\T) \leq k$.
Let us arbitrarily choose a node $r\in V(T)$ to be the root of $T$. 
We will show that, after this rooting, $\T$ can be transformed into a neat tree-cut decomposition of $G$.

Similarly to~\cite{Ganian0S15}, we will call a node $t\in V(T)\setminus \{r\}$ {\em bad}
if $|\adh(t\pi(t))|\leq 2$ and there exists a sibling $t'$ of $t$, such that there is an edge in $G$
between $X^T_{\ t\pi(t)}$ and $X^T_{\ t'\pi(t)}$.
Moreover,
for a bad vertex $t$ we say that a node $b$ is a {\em bad neighbor} of $t$ if
$b \in V(T_{t'\pi(t)})$ for some sibling $t'$ of $t$, and there is an edge between $X_b$ and $X^T_{\ t\pi(t)}$ in $G$.

We define the following two procedures, similarly to~\cite{Ganian0S15}, see Figure~\ref{fig:rerouting}:
\begin{quote}
{\sc Rerouting}$(t)$: let $t$ be a bad node and let $b$ be a bad neighbor of $t$ of maximum depth. Then remove the edge $t\pi(t)$ from $T$ and add a 
new edge $bt$, thus making $t$ a child of $b$.
\end{quote}

\begin{quote}
{\sc Top-down Rerouting}: as long as $(T,{\cal X})$ is not a neat tree-cut decomposition, pick a bad node $t$ of minimum depth and perform {\sc Rerouting}$(t)$.
\end{quote}

	\begin{figure}[H]
		\centering
		\vspace*{-10pt}
		{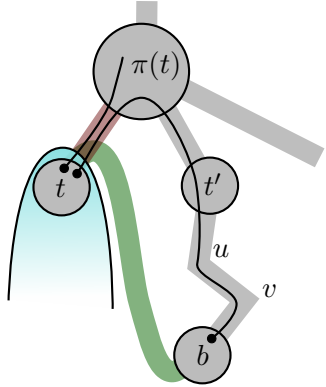}
		\caption{
			A decomposition with a bad node $t$ and a bad neighbor $b$ -- the rerouting procedure will reattach the subtree $T_{t\pi(t)}$ below $b$. Two edges of $G$ in $\adh(t t_p)$ are shown.}
		\label{fig:rerouting}
	\end{figure}

We first make sure that rerouting does not spoil the connectivity of a decomposition.

\begin{claim}
Let $\T=(T,{\cal X} = \{X_t:t \in V(T)\})$ be a connected rooted tree-cut decomposition and $t\in V(T)$ be a bad vertex of $T$.
If $\T'=(T',{\cal X})$ is the
rooted tree-cut decomposition obtained from $\T$ after running {\sc Rerouting}$(t)$, then $\T'$ is also connected.
\end{claim}

\begin{proof}
Notice first that, by construction, we have $E(T')=(E(T)\setminus \{t\pi(t)\})\cup\{tb\}$. Let $P$ denote the path in $T'$ (and in $T$) that leads from $\pi(t)$ to $b$.  

Observe that if $uv$ is an edge of $E(T')\setminus (E(P)\cup \{tb\})$,
then $V(T'_{uv})=V(T_{uv})$ and $V(T'_{vu})=V(T_{vu})$. Thus $G[X^T_{\ uv}]$ and $G[X^T_{\ vu}]$ are connected.
Consider now the edge $tb$ of $T'$. Notice that $V({T'}_{tb})=V(T_{t\pi(t)})$ and $V({T'}_{bt}) = V(T_{\pi(t)t})$.
This implies that the graphs $G[X^{T'}_{\ tb}]$ and $G[X^{T'}_{\ bt}]$ are connected as well.

Finally, consider an edge $e=uv$ of $P$ and, without loss of generality, we assume that $u$ is the parent of $v$. Notice then that
$$X^{T'}_{\ vu}=X^T_{\ vu}\cup X^T_{\ t\pi(t)}\qquad \text{and}\qquad X^{T'}_{\ uv}=X^T_{\ uv}\setminus X^T_{\ t\pi(t)}.$$ 
Recall that the subgraphs of $G$ induced by $X^T_{\ vu}$ and $X^T_{\ t\pi(t)}$ are connected.
Moreover, since $t$ is a bad vertex and $b$ is a bad neighbor of $t$, there exists an edge between $X^T_{\ t\pi(t)}$ and $X_{b}$ in $G$. 
Therefore, as $X_{b}\subseteq X^T_{\ vu}$, it follows that the graph $G[X^{T'}_{\ vu}]$ is connected.

To show that $X^{T'}_{\ uv}=X^T_{\ uv}\setminus X^T_{\ t\pi(t)}$ also induces a connected subgraph in $G$, recall that $\adh_{\T}(t\pi(t))$ is thin.
That is, $|\delta_{G}(X^T_{\ t\pi(t)})|\leq 2$, and recall that there is an edge between $X^T_{\ t\pi(t)}$ and $X_{b}$ in $G$. 
Observe then that $X^T_{\ uv}=X^{T'}_{\ uv}\cup X^T_{\ t\pi(t)}$, where there is at most one edge between $X^{T'}_{\ uv}$ and $X^T_{\ t\pi(t)}$.
Since $G[X^{T}_{\ uv}]$ is connected, this implies that after removing $X^T_{\ t\pi(t)}$ it remains connected, because any path that connected two vertices of $X^{T'}_{\ uv}$ and went through $X^T_{\ t\pi(t)}$ would imply at least two edges between these parts.
Thus, $G[X^{T'}_{\ uv}]$ is also connected.
\cqed\end{proof}

Next, we verify that rerouting does not increase the width.

\begin{claim}
Let $\T=(T,{\cal X})$ be a rooted tree-cut decomposition and $t\in V(T)$ be a bad vertex of $T$. If $\T'=(T',{\cal X})$ is the
rooted tree-cut decomposition obtained from $\T$ after running {\sc Rerouting}$(t)$, then $\width'(\T')\leq\width'(\T)$.
\end{claim}
\begin{proof}
As before, notice first that we have $E(T')=(E(T)\setminus \{t\pi(t)\})\cup\{tb\}$. Let $P$ denote the path in $T'$ (and in $T$) that leads from $\pi(t)$ to $b$.  

Observe that if $e$ is an edge of $E(T')\setminus (E(P)\cup \{tb\})$ 
then $\adh_{\T'}(e)=\adh_{\T}(e)$.
Notice also that $\adh_{\T'}(tb)=\adh_{\T}(t\pi(t))$.

Finally, let $e$ be an edge of $P$. Notice that the edge between $X^T_{\ t\pi(t)}$ and 
$X_{b}$ belongs to $\adh_{\T}(e)$ but not to $\adh_{\T'}(e)$. Moreover, since
$|\adh_{\T}(t\pi(t))|\leq 2$, there exists at most one edge that belongs to $\adh_{\T'}(e)$
but not to $\adh_{\T}(e)$. We conclude that $|\adh_{\T'}(e)|\leq |\adh_{\T}(e)|$, for every $e\in E(P)$.

In particular, $\max_{e\in E(T')} |\adh(e)| \leq \max_{e\in E(T)} |\adh(e)|$.
Furthermore, since the bag at every node is unchanged, $w_{\T'}(v)$ can be larger than $w_{\T}(v)$ only for $v=b$.
(Here, $w_{\T}(\cdot)$ and $w_{\T'}(\cdot)$ are functions $w(\cdot)$ as in the definition of $\width'$, applied to decompositions $\T$ and $\T'$, respectively.)
However, even in this case the only additional edge incident to $b$ is $tb$, whose adhesion has size at most~$2$.
We conclude that $w_{\T'}(v)\leq w_{\T}(v)$ for each $v\in V(T)$, and hence $\width'(\T')\leq \width'(\T).$
\cqed\end{proof}

The above two claims show that it is safe to apply the {\sc Rerouting} procedure. 
We now show that applying it exhastively, as described in procedure {\sc Top-down Rerouting}, always terminates within a polynomial number of steps.

\begin{claim} {\sc Top-down Rerouting} terminates after $\Oh(|T|^2)$ invocations of {\sc Rerouting}$(t)$.\end{claim}
\begin{proof}
We note  that the proof is again similar to the one in~\cite{Ganian0S15}.
However, we include it for the sake of completeness. 
For a tree-cut decomposition $\T=(T,{\cal X})$ and a node $v\in V(T)$
let $\depth(v,T)=\dist_{T}(v,r)$, where $r$ is the root of $T$. Notice, that for every $v\in V(T)$ we have $\depth(v,T)\leq |T|$,
and hence $$\sum_{v\in V(T)}\depth(v,T)\leq |T|^{2}.$$
Let $\T=(T,{\cal X})$ be a rooted
tree-cut decomposition of $G$ and $t$ be a bad node of $\T$ such that the distance of $t$ from $r$ is minimum. 
Let $\T'=(T',{\cal X})$ be the tree-cut decomposition of $G$ obtained after
performing {\sc Rerouting}$(t)$. Notice that,
for every $v\in V(T_{t\pi(t)})$, $\depth(v,T')\geq \depth(v,T)+1$.
Moreover, for every $v\in V(T_{\pi(t)t})$, $\depth(v,T')=\depth(v,T)$. This implies that
$$\sum_{v\in V(T)}\depth(v,T)<\sum_{v\in V(T')}\depth(v,T')\leq |T|^{2}.$$
Therefore, since the sum of depths of nodes increases at each step, and it is always upper bounded by $|T|^2$, the {\sc Top-down Rerouting} procedure terminates after at most $|T|^2$ steps.
\cqed\end{proof}

This concludes the proof that {\sc Top-down Rerouting} produces a neat tree-cut decomposition of $\width'$ bounded by $k^2+1$.
Since we always assume $|T|=\Oh(|G|)$, {\sc Rerouting} is invoked $\Oh(|G|^2)$ times. 
Finding a bad node and a bad neighbor can be done in $\Oh(|G|)$ time by inspecting edges of thin adhesions and computing the least-common-ancestor in $T$ of the two bags containing their endpoints, using e.g. Gabow and Tarjan's classical algorithm~\cite{GabowT83}.
Since {\sc Rerouting} can be performed in $\Oh(|G|)$ time, the algorithm runs in $\Oh(|G|^3)$ total time.
\end{proof}

We can now combine all tools developed so far to prove the following statement, which will later serve as an abstraction for getting tree-cut decompositions of $\F$-free graphs with good properties.

\begin{corollary}\label{corol:neat}
	Given an $\F$-free graph $G$, 
	a neat tree-cut decomposition of $\width'$ at most $b_\F$ of $G$
	can be computed in time $\Oh(\|G\| \cdot |G|^2)$,
	where $b_\F = 4(a_\F)^2 + 1$ is a constant depending on $\F$ only.
	Here, $a_\F$ is the constant given by Theorem~\ref{thm:Ffree_tctw}.
\end{corollary}
\begin{proof}
	We have $\tctw(G)\leq a_\F$ by Theorem~\ref{thm:Ffree_tctw}.
	By Theorem~\ref{thm:apx}, a tree-cut decomposition of $\width$ at most $2a_\F$ can be computed in time $\Oh(|G|^2)$.
	From it, by Theorem~\ref{thm:tcweqstcw'}, a decomposition of $\width'$ at most $2a_\F$ can be computed in time $\Oh(\|G\| \cdot |G|^2)$.
	Then, by Lemma~\ref{lem:conctvt}, a connected decomposition of $\width'$ at most $4(a_\F)^2 + 1$ can be computed in time $\Oh(\|G\|\cdot |G|^2)$.
	Finally, by Theorem~\ref{thm:neat}, a neat decomposition of $\width'$ at most $4(a_\F)^2 + 1$ can be computed in time $\Oh(|G|^3)$.
\end{proof}

Let $G$ be a graph with a neat tree-cut decomposition $\T=(T,{\cal X})$, and let $p \in V(T)$.
For $t \in N_T(p)$, we say the component $T_{tp}$ of $T-p$ is connected \emph{with a neat adhesion to $p$} if the adhesion of $tp$ is thin, and moreover all of its edges have an endpoint in $X_p$.
We now prove a result that shows what the neat decompositions are useful for: provided some node has many neighbors, all but a constant number of them is connected to it via neat adhesions.
\begin{corollary}\label{corol:neatAdhesions}
	Let $G$ be a graph with a neat tree-cut decomposition $\T=(T,{\cal X})$ with $\width'(\T)\leq b$, for some integer $b$.
	Then for every $p\in V(T)$, at most $2b+1$ of the connected components of $T-p$ are not connected with a neat adhesions to $p$.
\end{corollary}
\begin{proof}
	Let $p\in V(T)$.
	By the definition of $\width'$, at most $b$ of the edges in $\delta_T(p)$ have adhesions containing more than two edges of $G$; all other edges in $\delta_T(p)$ have thin adhesions.
	Additionally, at most one edge in $pt \in \delta_T(p)$ has the property that $\pi(p)=t$;
	all other edges $pt \in \delta_T(p)$ satisfy $\pi(t)=p$.	
	Consider then the remaining edges in $\delta_T(p)$, say $pt_1, \dots, pt_r$, for some $r\geq |\delta_T(p)|-b-1$. They have thin adhesions and satisfy $\pi(t_i)=p$.
	
	Recall that $\adh(p t_i)$ contains precisely the edges of $G$ with one endpoint in $X^T_{\ t_i p}$ and the other in $X^T_{\ p t_i}$.
	By definition of a neat decomposition (and since $\adh(t_i p) = \adh(t_i \pi(t_i))$ is thin), 
	the edges of $G$ contained in $\adh(p t_i)$ cannot have an endpoint in $X^T_{\ t'p}$ for any sibling $t'$ of $t_i$.
	This means that they have one endpoint in $X^T_{\ t_i p}$ and one in $X^T_{\ \pi(p) p} \cup X_p$
	(if $p$ is the root, assume $X^T_{\ \pi(p) p} = \emptyset$).
	However, the number of edges between $X^T_{\ t_i p}$  and $X^T_{\ \pi(p) p}$ is bounded by $|\adh(\pi(p) p)| \leq b$.
	Therefore, at least $r-b$ of the decomposition edges $t_i p$ ($i\in [r]$) have adhesions containing only edges of $G$ that have an endpoint in $X_p$.
	This means that for at least $r-b \geq |\delta_T(p)|-2b-1$ indices $i\in [r]$, the component $T_{t_i p}$ of $T-p$ is connected with a neat adhesion to $p$.
\end{proof}

\section{Protrusions}\label{sec:protrusions}
We now introduce the notion of a protrusion that is suitable for our problem. Namely, protrusions are $\F$-free parts of the graph with a constant-size boundary.

\begin{definition}\label{def:protrusion}
	An $r$-\emph{protrusion} of a graph $G$ is a set $X\subseteq V(G)$ such that $|\delta(X)|\leq r$ and $G[X]$ is $\F$-free.
\end{definition}

Recall that by Corollary~\ref{corol:neat}, the subgraph induced by a protrusion, as an $\F$-free graph, always has a neat tree-cut decomposition of $\width'$ bounded by a constant $b_\F$.
In the sequel, we will only deal with $2b_\F$- and $2$-protrusions.

\subsection{Replacing protrusions}

As in~\cite{FominLMS12}, the base for our kernelization algorithm is {\em{protrusion replacement}}. 
That is, the algorithm iteratively finds a protrusion $X$ that is large but has small $\delta(X)$, and replaces it with a gadget $X'$ that has the same behaviour, but is smaller.
The following lemma, whose proof is the main goal of this section, formalizes this intuition.

\begin{lemma}\label{lem:basicReplacer}
There is a constant $c_\F$ and algorithm that, given a graph $G$ and a $2b_\F$-protrusion $X$ in it with $\|G[X]\| > c_\F$, outputs in linear time a graph $G'$ with $\OPT(G)=\OPT(G')$ and $\|G'\| < \|G\|$.

Moreover, there is a linear-time algorithm working as follows:
given a subset $F'$ of edges of $G'$ such that $G'-F'$ is $\F$-free, the algorithm computes a subset $F$ of edges of $G$ such that $G-F$ is $\F$-free and $|F|\leq |F'|$ (and is called a solution-lifting algorithm).
\end{lemma}

The proof of Lemma~\ref{lem:basicReplacer} follows closely the strategy used by Fomin et al.~\cite{FominLMS12}: 
Every $2b_\F$-protrusion can be assigned a type, where the number of types is bounded by a function depending on $\F$ only.
The type of a protrusion can be computed efficiently due to protrusions having constant treewidth.
Protrusions with the same type behave in the same way with respect to the problem of our interest, and hence can be replaced by one another.
Therefore, we store a replacement table consisting of the smallest protrusion of each type, so that every larger protrusion can be replaced by a smaller representative stored in the table.
The lifting algorithm finds, using dynamic programming, 
a partial solution in the large protrusion that has the same behaviour as the given partial solution in the replacement protrusion, while being not larger. 

We now proceed with implementing this plan formally.
We start with defining {\em{boundaried graphs}}.

\newcommand{\Gf}{\mathbb{G}}
\newcommand{\Hf}{\mathbb{H}}

\begin{definition}
An {\em{$r$-boundaried graph}} consists of an underlying graph $G$ and an $r$-tuple $(u_1,\ldots,u_r)$ of (not necessarily different) vertices of $G$, called the {\em{boundary}}. 
Given two $r$-boundaried graphs 
$$\Gf=(G,(u_1,\ldots,u_r))\qquad \textrm{and} \qquad \Hf=(H,(v_1,\ldots,v_r)),$$
we define their {\em{gluing}}, denoted $\Gf\oplus\Hf$, to be the following graph: take the disjoint union of $G$ and $H$, and for each $i\in [r]$ add one edge $u_iv_i$.
Finally, we define $\|\Gf\|$ to be $\|G\|$.
\end{definition}

We extend all notation for graphs to boundaried graphs, always applying it to the underlying graph. Thus, we can talk about, e.g., $\F$-free boundaried graphs.

Boundaried graphs can be naturally equipped with a Myhill-Nerode-like equivalence relation concerning the problem of our interest.

\begin{definition}
Two $r$-boundaried graphs $\Gf_1$ and $\Gf_2$ are called {\em{$\F$-equivalent}} if for every $r$-boundaried graph $\Hf$, the following holds:
$$\OPT(\Gf_1\oplus \Hf)=\OPT(\Gf_2\oplus \Hf).$$
\end{definition}

Obviously, $\F$-equivalence is an equivalence relation on $r$-boundaried graphs. 
We now introduce a condition that implies $\F$-equivalence, which will be combinatorially easier to handle.

Suppose $\Gf$ is an $r$-boundaried graph, with boundary $(u_1,u_2,\ldots,u_r)$. 
Define the {\em{extended graph}} $\extnd{\Gf}$ as follows: for each $i\in [r]$, introduce a new vertex $\cop{u_i}$ that is adjacent only to $u_i$.
Suppose further that $Q$ is some graph. 
If $\phi$ is a partial function from $V(Q)$ to $[r]$, then by a {\em{$\phi$-rooted immersion model}} of $Q$ in $\extnd{\Gf}$ 
we mean an immersion model of $Q$ in $\extnd{\Gf}$ that is faithful w.r.t. $\phi$ in the following sense: 
for each vertex $v$ of $Q$ that has defined image under $\phi$, $v$ is mapped to $\cop{u_{\phi(v)}}$ in the immersion model.

Fix a positive integer $r$ and recall that $\MX=\max_{H\in \F} \|H\|$. Consider a graph $Q$ and a partial function $\phi$ from $V(Q)$ to $[r]$. 
We call the pair $(Q,\phi)$ {\em{relevant}} if the following conditions hold:
\begin{itemize}
\item $\|Q\|\leq (r+1)\MX$ and $Q$ has no isolated vertices; and
\item $\phi$ is non-empty, i.e., it assigns a value to at least one argument.
\end{itemize}
The set of relevant pairs will be denoted by $\Rr_{r,\F}$. Observe that 
\begin{equation}\label{eq:rel-pairs}
|\Rr_{r,\F}|\leq 2^{\poly(r,\MX)}.
\end{equation}
Indeed, there are at most $2^{\poly(r,\MX)}$ graphs with at most $(r+1)\cdot \MX$ edges and no isolated vertices, and 
for each of them there are at most $(r+1)^{\Oh((r+1)\MX)}$ possible partial functions $\phi$.

\begin{definition}
Let $r$ be a positive integer and let $\Gf$ be an $r$-boundaried graph.
For a set $\Ss\subseteq \Rr_{r,\F}$ of relevant pairs, the {\em{deletion number of $\Gf$ w.r.t. $\Ss$}} is the minimum number of edges that need to be deleted from $\extnd{\Gf}$ 
so that it does not admit a $\phi$-rooted immersion model of $Q$, for each $(Q,\phi)\in \Ss$. 
Note that the edges $u_i\cop{u_i}$, for $i\in [r]$, may also be deleted in this definition.
The {\em{signature}} of $\Gf$, denoted $\sigma[\Gf]$, is the function from subsets of $\Rr_{r,\F}$ to nonnegative integers defined as follows:
$$\sigma[\Gf](\Ss)=\textrm{deletion number of $\Gf$ w.r.t. $\Ss$}.$$
\end{definition}

The following lemma explains the relation between $\F$-equivalence and signatures.

\begin{lemma}\label{lem:equiv-implication}
If two $\F$-free $r$-boundaried graphs have the same signatures, then they are $\F$-equivalent.
\end{lemma}
\begin{proof}
Let $\Gf_1,\Gf_2$ be a pair of $\F$-free $r$-boundaried graphs that have the same signature. For $t=1,2$, let $(u^t_1,\ldots,u^t_r)$  be the boundary of $\Gf_t$.
Take any $r$-boundaried graph $\Hf$, and let $(v_1,\ldots,v_r)$ be its boundary.

We need to prove that $\OPT(\Gf_1\oplus \Hf)=\OPT(\Gf_2\oplus \Hf)$.
It suffices to prove that $\OPT(\Gf_1\oplus \Hf)\leq \OPT(\Gf_2\oplus \Hf)$, because then the converse inequality will follow by symmetry.
Throughout the proof, we implicitly identify $\Gf_1$, $\Gf_2$, and $\Hf$ with their copies in the gluings $\Gf_1\oplus \Hf$ and $\Gf_2\oplus \Hf$.
We also use the extended graphs $\extnd{\Gf_1}$ and $\extnd{\Gf_2}$, with the notation $\cop{\cdot}$, and injective mappings 
\begin{eqnarray*}
\iota_1\colon E(\extnd{\Gf_1})\to E(\Gf_1\oplus \Hf)\\
\iota_2\colon E(\extnd{\Gf_2})\to E(\Gf_2\oplus \Hf)
\end{eqnarray*}
defined as follows. If $e\in E(\Gf_1)$, then $\iota_1(e)=e$, and if $e=u^1_i\cop{u^1_i}$ for some $i=1,\ldots,r$, then $\iota_1(e)=u^1_iv_i$. Mapping $\iota_2$ is defined in the same way.

Suppose $F_1$ is an optimum-size subset of edges of $\Gf_1\oplus\Hf$ such that $(\Gf_1\oplus\Hf)-F_1$ is $\F$-free; that is, $|F_1|=\OPT(\Gf_1\oplus \Hf)$.
Let $L_1=\iota_1^{-1}(F_1\setminus E(\Hf))$; that is, $L_1$ consists of all edges of $\extnd{\Gf_1}$ that correspond to edges of $F_1$ under mapping $\iota_1$.
Let $\Ss$ be the set of all relevant pairs $(Q,\phi)\in \Rr_{r,\F}$ for which $\extnd{\Gf_1}-L_1$ does not admit a $\phi$-rooted immersion model of $Q$.
Since $\Gf_1$ and $\Gf_2$ have the same signatures, there is a subset of edges $L_2\subseteq E(\extnd{\Gf_2})$ with $|L_2|\leq |L_1|$ such that 
$\extnd{\Gf_2}-L_2$ also does not admit a $\phi$-rooted immersion model of $Q$, for every $(Q,\phi)\in \Ss$.
We define $F_2\subseteq E(\Gf_2\oplus\Hf)$ as follows:
$$F_2=\iota_2(L_2)\cup (F_1\cap E(\Hf)).$$
Since $|L_2|\leq |L_1|$, we also have that $|F_2|\leq |F_1|$. Hence it suffices to show that $(\Gf_2\oplus\Hf)-F_2$ is $\F$-free.

For the sake of contradiction suppose that $(\Gf_2\oplus\Hf)-F_2$ contains an immersion model $\Ii$ of some graph $H\in \F$.
Clearly $\Ii$ must use at least one edge outside $E(\Hf)$, because otherwise $\Ii$ would be also an immersion model of $H$ in $(\Gf_1\oplus\Hf)-F_1$, which is $\F$-free by assumption.
Also, $\Ii$ must use at least one edge outside $E(\Gf_2)$, because otherwise it would be an immersion model of $H$ in $\Gf_2$, which is $\F$-free by the supposition of the lemma.

Take any edge $e$ of $H$, and let $P_e$ be the path in the model $\Ii$ that is the image of $e$. 
Each vertex traversed by $P_e$ belongs either to $V(\Gf_2)$ or to $V(\Hf)$.
For each maximal interval $I$ on $P_e$ of vertices belonging to $V(\Gf_2)$, consider the path in $\extnd{\Gf_2}$ constructed as follows:
take all edges of $P_e$ incident to the vertices of $I$ (so including the edge preceding and succeeding $I$ on the path), and map them to the edges of $\extnd{\Gf_2}$ using $\iota_2^{-1}$.
This image is a path in $\extnd{\Gf_2}$ whose endpoints are either copies $\cop{u^2_i}$ of some boundary vertices, or the original endpoints of $P_e$.

Starting from $H$, construct a graph $Q$ as follows (see Fig.~\ref{fig:partial-obs} for reference). The vertex set of $Q$ consists of all the vertices of $H$ that are mapped to $V(\Gf_2)$ in the model $\Ii$,
plus the set of all the vertices $\cop{u^2_i}$ for which the edge $u^2_iv_i$ is used in the model $\Ii$.
The edges of $Q$ are defined by the construction of the previous paragraph: every path $R$ constructed for some maximal interval on some path $P_e$ gives rise to an edge in $Q$ connecting the endpoints of $R$.
It is now easy to see that the above paths define a $\phi$-rooted immersion model of $Q$ in $\extnd{\Gf_2}-L_2$, where $\phi$ assigns each vertex $\cop{u^2_i}$ its index $i$.

\begin{figure}
    \centering
\scalebox{.89}{\def\svgwidth{\textwidth}\input{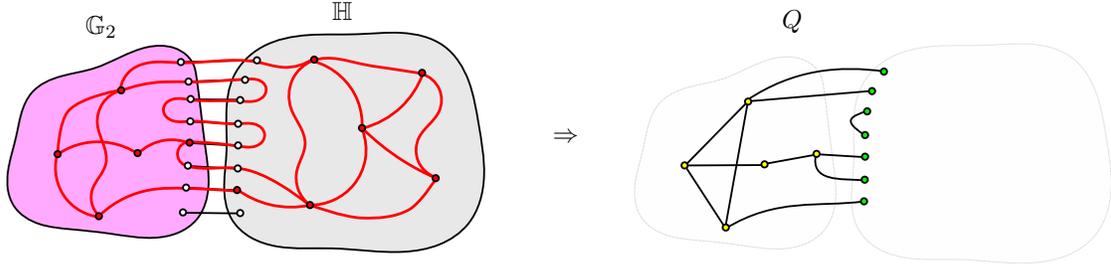}}
    \caption{Construction of graph $Q$ from the immersion model $\Ii$. The model $\Ii$ is depicted on the left panel; the red vertices are the images of the vertices of $H$.
    The obtained graph $Q$ is on the left panel. The yellow vertices are the images of vertices of $H$ that lie within $V(\Gf_2)$, whereas the green vertices are the copies of
    boundary vertices that are included in the vertex set of $Q$. The former graphs $\Gf_2$ and $\Hf$ are depicted in very light grey in order to show from where the different parts of $Q$ come from.}\label{fig:partial-obs}
\end{figure}

We now verify that $(Q,\phi)$ is a relevant pair. First, since every graph of $\F$ is connected and has at least one edge, it is immediate that $Q$ has no isolated vertices.
For every edge $e$ of $H$, the path $P_e$ can alternate between $V(\Gf_2)$ and $V(\Hf)$ at most $r$ times, and hence $e$ can give rise to at most $r+1$ edges in $Q$; it follows
that $$\|Q\|\leq (r+1)\|H\|\leq (r+1)\MX.$$ Finally, since $\Ii$ uses at least one edge outside $E(\Hf)$ and at least one edge outside $E(\Gf_2)$, we conclude that neither $Q$ nor $\phi$ is empty.

Since $(Q,\phi)$ is a relevant pair for which there is a $\phi$-rooted immersion model of $Q$ in $\extnd{\Gf_2}-L_2$, we have that $(Q,\phi)\notin \Ss$.
By the way  we defined $\Ss$, it follows that there is a $\phi$-rooted immersion model of $Q$ in $\extnd{\Gf_1}-L_1$.
Take the edges of this model, map them according to $\iota_1$ to edges of $\Gf_1\oplus \Hf$, and add all the edges used by model $\Ii$ within $E(\Hf)$. 
It can be now easily seen that all these edges form an immersion model of $H$ in $(\Gf_1\oplus \Hf)-F_1$, which is a contradiction with $(\Gf_1\oplus \Hf)-F_1$ being $\F$-free.
\end{proof}

It is not hard to see that the deletion numbers in fact cannot be too large.

\begin{lemma}\label{lem:bounded-deletion-number}
If $\Gf$ is an $r$-boundaried graph and $\Ss\subseteq \Rr_{r,\F}$ is a subset of relevant pairs, then the deletion number of $\Gf$ w.r.t. $\Ss$ is at most $r$.
\end{lemma}
\begin{proof}
Let $(u_1,\ldots,u_r)$ be the boundary of $\Gf$.
Since $\phi$ is non-empty for each $(Q,\phi)\in \Ss$, in order to make $\extnd{\Gf}$ not admit $\phi$-rooted minor of $Q$ one can always remove all the edges $u_i\cop{u_i}$, for $i\in [r]$.
Hence, the deletion number of $\Gf$ w.r.t. $\Ss$ is upper bounded by the number of these edges, that is, by~$r$.
\end{proof}

Lemmas~\ref{lem:equiv-implication} and~\ref{lem:bounded-deletion-number}, together with~\eqref{eq:rel-pairs}, immediately yield the following.

\begin{corollary}\label{cor:num-sig}
The number of possible signatures of $r$-boundaried graphs is at most $2^{2^{2^{\poly(r,\MX)}}}$. Consequently, $\F$-equivalence has at most this many equivalence classes.
\end{corollary}

Finally, we need the algorithmic tractability of signatures.

\begin{lemma}\label{lem:comp-sig}
For every positive integer $r$, there exists a linear-time algorithm that, given an $\F$-free $r$-boundaried graph $\Gf$, computes its signature.
\end{lemma}
\begin{proof}
For each such subset $\Ss$ of relevant pairs, the deletion number of $\Gf$ w.r.t. $\Ss$ can be computed in linear time as follows.
First, observe that, due to $\Gf$ being $\F$-free, by Proposition~\ref{lem:twBound} and Theorem~\ref{thm:Ffree_tctw} we infer that the treewidth of $\extnd{\Gf}$ is bounded by a constant depending on $\F$ only.
Hence, using Bodlaender's algorithm~\cite{Bodlaender96} we can compute in linear time a tree decomposition of $\extnd{\Gf_1}$ of constant width.
Then, on this tree decomposition we apply the optimization variant of Courcelle's theorem, due to Arnborg et al.~\cite{ArnborgLS91} (see also~\cite[Theorem 7.12]{platypus} for a modern presentation).
For this, we observe that finding the minimum cardinality of an edge subset of $\extnd{\Gf}$ that hits all $\phi$-rooted immersion model of $Q$, for each $(Q,\phi)\in \Ss$, can be
expressed in a straightforward way as an $\mathbf{MSO}_2$ optimization problem; the formula's length depends only on $r$ and $\F$.
Thus, the algorithm of Arnborg et al.~\cite{ArnborgLS91} solves this optimization problem in linear-time, yielding the deletion number of $\Gf$ w.r.t. $\Ss$.
By applying this procedure to all subsets $\Ss$ of relevant pairs, whose number is bounded by a constant depending on $r$ and $\F$, we obtain the whole signature of $\Gf$. 
\end{proof}

We are ready to prove the basic protrusion replacement lemma, i.e., Lemma~\ref{lem:basicReplacer}.

\begin{proof}[Proof of Lemma~\ref{lem:basicReplacer}]
We first describe the algorithm that computes $G'$.
Recall that, by Corollary~\ref{cor:num-sig}, for any $r\leq 2b_\F$ the number of possible signatures of $r$-boundaried graphs is bounded by a constant depending $\F$ only.
Define the following table $T$: 
for each $r\leq 2b_\F$ and each possible signature $\rho$ of $r$-boundaried graphs, we store in $T$ the smallest, in terms of the number of edges, $\F$-free $r$-boundaried graph $\Gf_\rho$ for which $\sigma[\Gf_\rho]=\rho$.
If no $\F$-free $r$-boundaried graph has signature $\rho$, a marker $\bot$ is stored instead.
Note that table $T$ depends only on family $\F$, and hence can be hardcoded in the algorithm.
Define $c_\F$ to be the largest number of edges among the graphs stored in $T$; then $c_\F$ is a constant depending on $\F$ only.

Let $r=|\delta(X)|$. Based on $G[X]$ and $G-X$, define $r$-boundaried graphs $\Gf_X$ and $\Hf$ as follows. 
The underlying graph of $\Gf_X$ is $G[X]$, and of $\Hf$ is $G-X$. 
Fix an arbitrary ordering $e_1,e_2,\ldots,e_r$ of the edges of $\delta(X)$.
Then the $i$-th boundary vertex of $\Gf_X$ is the endpoint of $e_i$ that lies in $X$, and the $i$-th boundary vertex of $\Hf$ is the second endpoint of $e_i$, the one that lies outside of $X$.
It follows that $G=\Gf_X\oplus \Hf$.

Using the algorithm of Lemma~\ref{lem:comp-sig}, compute the signature $\rho:=\sigma[\Gf_X]$.
Note here that $\Gf_X$ is $\F$-free by the supposition that $X$ is a protrusion.
Since $\Gf_X$ has signature $\rho$, it follows that table $T$ stores some $r$-boundaried graph $\Gf_{\rho}$ with the same signature.
As $\|G[X]\|=\|\Gf_X\|>c_\F$ and $\|\Gf_{\rho}\|\leq c_\F$, we have that $\|\Gf_{\rho}\|< \|\Gf_X\|$.

Define
$$G':=\Gf_{\rho}\oplus \Hf.$$
As $\|\Gf_{\rho}\|<\|\Gf_X\|$, we have that $\|G'\|<\|G\|$. 
Since $\Gf_X$ and $\Gf_\rho$ have the same signatures and are both $\F$-free, by Lemma~\ref{lem:equiv-implication} we have that they are $\F$-equivalent.
Hence
$$\OPT(G)=\OPT(\Gf_X\oplus \Hf)=\OPT(\Gf_\rho\oplus \Hf)=\OPT(G'),$$
and we conclude that $G'$ can be output by the algorithm.

\smallskip

We now describe the solution lifting algorithm.
Suppose we are given a subset $F'$ of edges of $G'$ such that $G'-F'$ is $\F$-free.
Recall that $G'=\Gf_\rho\oplus \Hf$.
Let $F'_\Hf=F'\cap E(\Hf)$ and let $F'_\rho=F'\setminus F'_\Hf$; here, we implicitly identify $\Gf_\rho$ and $\Hf$ with their copies in the gluing $G'=\Gf_\rho\oplus \Hf$.
Consider the extended graph $\extnd{\Gf_\rho}$, and let $\widetilde{F'_\rho}$ be the image of $F'_\rho$ under the mapping $\iota^{-1}$ defined as in the proof of Lemma~\ref{lem:equiv-implication}:
the edges of $\Gf_\rho$ are mapped to themselves, while the edges between $\Gf_X$ and $\Hf$ are mapped to the corresponding edges between the boundary vertices and their copies in  $\extnd{\Gf_\rho}$.

Since $\extnd{\Gf_\rho}$ is a graph of constant size, we can compute in constant time the subset $\Ss\subseteq \Rr_{r,\F}$ of those relevant pairs $(Q,\phi)\in \Rr_{r,\F}$,
for which $\extnd{\Gf_\rho}$ does not admit a $\phi$-rooted immersion model of $Q$.
Observe that since the signatures of $\Gf_\rho$ and $\Gf_X$ are the same, the deletion numbers of $\Gf_\rho$ and $\Gf_X$ w.r.t $\Ss$ are equal.
Hence, there exists a subset $\widetilde{F_X}$ of edges of $\extnd{\Gf_X}$ with $|\widetilde{F_X}|\leq |\widetilde{F'_\rho}|$, such that also in $\extnd{\Gf_X}-\widetilde{F_X}$
there is no $\phi$-rooted immersion model of $Q$, for each $(Q,\phi)\in \Ss$.

Observe that such set $\widetilde{F_X}$ can be computed in linear time using the algorithm of Arnborg et al.~\cite{ArnborgLS91} as follows.
Just as in Lemma~\ref{lem:comp-sig}, the fact that $\Gf_X$ is $\F$-free implies that $\extnd{\Gf_X}$ has constant treewidth.
Hence, we can compute its tree decomposition of constant width using Bodlaender's algorithm~\cite{Bodlaender96}; this takes linear time.
Then, on this decomposition we run the algorithm of Arnborg et al.~\cite{ArnborgLS91} for the $\mathbf{MSO}_2$ optimization problem defined as follows:
find the smallest subset of edges whose removal leaves no $\phi$-rooted immersion model of $Q$, for each $(Q,\phi)\in \Ss$.
The algorithm of Arnborg et al.~\cite{ArnborgLS91} can within the same linear running time also reconstruct the solution, so we are indeed able to construct $\widetilde{F_X}$.

Let now $F_X\subseteq E(G)$ be the image of $\widetilde{F_X}$ under the mapping $\iota$ defined as in the proof of Lemma~\ref{lem:equiv-implication}:
the edges of $\Gf_X$ are mapped to themselves, while the edges between the boundary vertices and their copies are mapped to the corresponding edges between $\Gf_X$ and $\Hf$.
Define $F:=F_X\cup F'_\Hf$.
By the construction of $F_X$ we have that for every relevant pair $(Q,\phi)\in \Rr_{r,\F}$, if $\extnd{\Gf_X}-\widetilde{F_X}$ admits a $\phi$-rooted immersion model of $Q$, then so does $\extnd{\Gf_\rho}-\widetilde{F'_\rho}$.
A simple replacement argument, essentially the same as in the proof of Lemma~\ref{lem:equiv-implication}, shows that the $\F$-freeness of $G'-F'$ implies that $G-F$ is also $\F$-free.
Also, $|F|\leq |F'|$ due to $|\widetilde{F_X}|\leq |\widetilde{F'_\rho}|$, so the solution $F$ can be returned by the solution-lifting algorithm.
\end{proof}

We henceforth define a \emph{replaceable protrusion} in $G$ as a $2b_\F$-protrusion $X$ with $\|G[X]\|>c_\F$, where $c_\F$ is the constant given by Lemma~\ref{lem:basicReplacer}.

\smallskip

Note that the proof of Lemma~\ref{lem:basicReplacer} a priori does not give any concrete upper bound on the sizes of the replacement graphs stored in table $T$, and on the obtained constant $c_\F$.
Also, it is unclear how to compute the table $T$ based on the knowledge of $\F$.
Obviously, we do not need the computability of $T$ or any concrete upper bound on $c_\F$, because we design algorithms for a family $\F$ fixed in advance, so objects
that depend on $\F$ only may be hard-coded in the algorithms.
However, we find it instructive to discuss the matter of computability of $T$ and the bound on $c_\F$, at least intuitively.

First of all, the algorithm of Arnborg et al.~\cite{ArnborgLS91} is based on constructing an automaton that traverses the given tree decomposition of a graph.
In our case, we have a constant upper bound on the minimum size of the sought set, so we are actually working with a finite-state tree automaton.
By tracing the standard translation from $\mathbf{MSO}_2$ to tree automata, one can estimate the number of states in the tree automaton constructed by the algorithm.
This number is bounded by a tower function of constant height applied to $\MX$ and the width of the decomposition; this is because the formula expressing the problem has constant quantifier rank. 
This also gives an upper bound on the minimum size of a tree decomposition that the automaton evaluates to a given state, which directly corresponds to the sizes of graphs stored in table $T$.
The automaton can be explicitly constructed by the algorithm of Arnborg et al.~\cite{ArnborgLS91}, and from the automaton one can retrieve small candidates for graphs stored in $T$.

The above strategy roughly shows that the graphs that are stored in $T$ are of size bounded by a tower function of constant height applied to $\MX$ and the width of the decomposition.
However, based on this idea one can also give a direct proof, as follows.
Take any graph $\Gf$ stored in $T$, and let $\T$ be its tree-cut decomposition of width at most $b_\F$.
With every node $x$ of $\T$ one can associate a boundaried graph $\Gf_x$, which corresponds to the subtree rooted at $x$; the adhesion between $x$ and its parent forms the boundary.
Assume for a moment that $\Gf$ is very large.
Suppose first that the depth of $\T$ is large, more precisely larger than the total number of signatures of $r$-boundaried graphs, for $r\leq b_\F$.
Then there is some root-to-leaf path in $\T$ that contains two nodes $x$ and $y$, say $x$ being an ancestor of $y$, for which $\Gf_x$ and $\Gf_y$ have the same signatures.
It can be easily seen that the part between $\Gf_x$ and $\Gf_y$ can be ``unpumped'': we can replace $\Gf_x$ with $\Gf_y$, obtaining a smaller graph $\Gf'$ with the same signature as $\Gf$.
If this unpumping cannot be applied, then the depth of $\T$ is bounded by the number of signatures, and $\Gf_x$ can be large only if some node has a large number of children.
But then again, a similar unpumping strategy can be applied if the number of children is larger than some constant depending on the number of possible signatures. 
Thus we obtain an explicit upper bound on the size of a graph that can be stored in $T$ instead of $\Gf$.

Once all these arguments are formalized, one can prove the following result that gives an upper bound on the sizes of graphs that are stored in $T$.

\begin{lemma}\label{lem:table-bound}
Suppose $r$ is a positive integer and $\rho$ is a signature of $r$-boundaried graphs. If there exists an $\F$-free $r$-boundaried graph with signature $\rho$, then there is also one with at most 
$\mathsf{4exp}(\poly(r,\MX))$ vertices and edges, where $\mathsf{4exp}(\cdot)$ is the $4$-times folded exponential function.
\end{lemma}

Note that once we have a computable upper bound, table $T$ may be constructed from $\F$ in constant time by brute force.
We would like to remark that the same unpumping strategy was recently applied by Chatzidimitriou et al.~\cite{Chatzidimitriou2015logopt} 
for the parameter {\em{tree-partition width}}, which is similar to treewidth.
The full proof of Lemma~\ref{lem:table-bound} will appear in the journal version of this paper.

\subsection{Finding excessive protrusions}

Recall that a replaceable protrusion in a graph $G$ is a $2b_\F$-protrusion $X$ with $\|G[X]\| > c_\F$.
To find replaceable protrusions in the input graph, we need to assume some additional connectivity constraint (which will be implied from a connected tree-cut decomposition) -- this is captured by the following definition.
The larger protrusion  size is needed to make any connected component of the protrusion replaceable.

\begin{definition}
	A $2b_\F$-protrusion $B$ in a connected graph $G$ is called {\em{excessive}} if $\|G[B]\|>2 b_\F \cdot c_\F$ and $G-B$ has at most two connected components.
\end{definition}

Replaceable protrusions could be found easily if we allowed a (far worse) running time of the form $\|G\|^{\Oh(b_\F)}$, but this would affect the running times in both our main results.
With the above definition in hand, we use the following two techniques instead.

The first is \emph{important cuts}, introduced by Marx~\cite{Marx06}, see also the exposition in \cite[Chapter~8.2]{platypus}.
Intuitively, we consider ($S,T$)-cuts (i.e., edge sets whose removal separates the vertex sets $S$ and $T$) that are `pushed' towards $T$, meaning that we make the set of vertices reachable from $S$ inclusion-wise maximal, without increasing the cut size.
Such cuts can be effectively enumerated, allowing us to find a protrusion's boundary.

\begin{definition}\label{def:importantCut}
	Consider a graph $G$ and disjoint vertex sets $S,T\subseteq V(G)$. Let $\Delta \subseteq E(G)$ be an $(S,T)$-cut and let $R$ be the set of vertices reachable from $S$ in $G-\Delta$.
	We say $\Delta$ is an \emph{important cut} if it is inclusion-wise minimal and there is no $(S,T)$-cut $\Delta'$ with $|\Delta'|\leq |\Delta|$ such that $R' \supset R$, where $R'$ is the set of vertices reachable from $S$ in $G-\Delta'$.
\end{definition}

\begin{lemma}[\cite{Marx06}]\label{lem:importantCut}
	Let $S,T\subseteq V(G)$ be two disjoint sets of vertices in a graph $G$ and let $k\geq 0$.
	The set of all important $(S,T)$-cuts of size at most $k$ can be enumerated in time $\Oh(4^k \cdot k \cdot \|G\|)$.
\end{lemma}

The second technique we use is \emph{randomized contractions} by Chitnis et al.~\cite{ChitnisCHPP12}.
While \emph{randomized} refers to the intuition behind this technique, following~\cite{ChitnisCHPP12} we use the technique of \emph{splitters} of Naor et al.~\cite{NaorSS95} 
to make its usage deterministic. A convenient black-box access to splitters is given by the following lemma.

\begin{lemma}[\cite{ChitnisCHPP12}]\label{lem:splitters}
	Given a set $U$ of size $m$ together with integers $0\leq a,b \leq m$, one can in
	time $2^{\Oh(\min(a,b) \log(a+b))}\cdot m \log m$  construct a family $\splitterF$ of at most $2^{\Oh(\min(a,b) \log(a+b))}\cdot \log m$ subsets of $U$, such that the following holds: 
	for any sets $A, B \subseteq U$ with $A\cap B = \emptyset$, $|A|\leq a$, $|B|\leq b$, there exists a set $T \in \splitterF$ with $A \subseteq T$ and $B\cap T = \emptyset$.
\end{lemma}

These two techniques allow us to reduce excessive protrusions: we use the randomized contractions technique to find a large enough subset of a presumed excessive protrusion, after which important cuts allow us to find a boundary that makes this subset a replaceable protrusion.

\begin{lemma}\label{lem:randomizedContractions}
	There is an algorithm that, given a connected graph $G$, runs in time $\Oh(\|G\| \log \|G\| \cdot |G|^2)$ and either correctly concludes that $G$ does not contain any excessive protrusion,
	or it outputs some replaceable protrusion in $G$.
\end{lemma}
\begin{proof}
	We describe the algorithm under the assumption that $G$ contains some excessive protrusion; in this case, we show that the algorithm can compute some replaceable protrusion.
	If the algorithm fails to find some replaceable protrusion, then this certifies that $G$ has no excessive protrusions, and this conclusion can be reported by the algorithm.
	
	Let $B$ be an excessive protrusion in $G$.
	Since $G$ is connected and $B$ is a $2b_\F$-protrusion, $B$ induces at most $2b_\F$ connected components in $G$. Let $B'$ be the largest one (in the number of edges).
	Then clearly $B'$ is a $2b_\F$-protrusion with $\|B'\| > c_\F$ and with $G[B']$ connected.
	Furthermore, $G-B'$ has at most two components, because $G-B$ has, and every connected component of $G[B]$ is adjacent to at least one of the components of $G-B$, due to the connectivity of $G$.
	We consider $B'$ instead of $B$ from now on.
	
	Let $T$ be a tree spanning a subset of $B'$ with $\min(|B'|, c_\F + 2)$ vertices.
	Then $\|T\| \leq c_\F + 1$ and $\|G[V(T)]\| > c_\F$.
	Let $s_1, s_2$ be arbitrary vertices in the two components of $G-B'$ (set $s_1 = s_2$ if it has only one component) and set $S = \{s_1, s_2\}$.
	
	To find $T$, we now apply Lemma~\ref{lem:splitters} for universe $U := E(G)$ and constants $a:=\|T\| \leq c_\F + 1$ and $b:=|\delta(B')|\leq 2b_\F$.
	Thus, in time $\Oh(\|G\| \log \|G\|)$ we construct a family $\splitterF$ of $\Oh(\log \|G\|)$ subsets of $E(G)$ with the following guarantee:
	for at least one $F \in \splitterF$, we have $E(T) \subseteq F$ and $\delta(B') \cap F = \emptyset$.
	The algorithm guesses this set $F \in \splitterF$ and the vertices of $S$ (by iterating over $|\splitterF| \cdot |G|^2$ possibilities); we shall consider the guess \emph{successful} if $F$ indeed has the above property and $S$ indeed intersects each component of $G-B'$.
	
	Make the edges of $F$ undeletable by considering the graph $\bar{G}$ obtained from $G$ by contracting all edges in $F$ (we use the same vertex labels in $\bar{G}$ by abuse of notation).
	Observe that $\delta(B')$ is an $(S,V(T))$-cut in $G$ of size at most $2b_\F$.
	If the guess was successful, it is an $(S,V(T))$-cut  of size at most $2b_\F$ in $\bar{G}$ too, and furthermore by choice of $S$, the set of vertices reachable from $S$ in $\bar{G}-\delta(B')$ is precisely $V(\bar{G})\setminus B'$.
	
	Consider a corresponding important cut, that is, let $\Delta\subseteq E(\bar{G})$ be an important $(S,V(T))$-cut of size at most $2b_\F$ such that the set of vertices reachable from $S$ in $\bar{G}-\Delta$ contains $V(\bar{G})\setminus B'$ (the existence of such a cut is easily proved, see~\cite{Marx06,platypus}).
	Let $\bar{X}$ be the set of vertices reachable from $T$ in $\bar{G}-\Delta$; then $\bar{X} \subseteq B'$ and $\delta(\bar{X}) \subseteq \Delta$.
	
	Let $X$ be the set of vertices in $G$ that gets contracted to $\bar{X}$ in $\bar{G}$.
	Then also $X\subseteq B'$ and $\delta(\bar{X}) \subseteq \Delta$ (as a subset of $E(G)\setminus F$).	
	That is, $X$ is $\F$-free (because $B'$ is) and $|\delta(X)|\leq 2b_\F$, meaning $X$ is a $2b_\F$-protrusion.
	As $X$ contains $V(T)$, we have $\|G[X]\| \geq \|G[V(T)]\| > c_\F$, meaning $X$ is a replaceable protrusion.

	Since $\Delta$ is an important cut of size at most $2b_\F$, we can use Lemma~\ref{lem:importantCut} to find it, and thus to find $X$, in $\Oh(\|G\|)$ time.
	Therefore, for at least one of $\Oh(|G|^2 \log \|G\|)$ guesses, the algorithm will find a replaceable protrusion.
	To handle unsuccessful guesses, the algorithm checks if the obtained set $X$ is in fact a replaceable protrusion; this takes $\Oh(\|G\|)$ time for each guess, by Proposition~\ref{lem:ffreeFpt}.
\end{proof}

We remark that we only defined excessive protrusions in connected graphs.
Note that if $B$ is an excessive protrusion in a connected component $H$ of $G$, it would not necessarily be an excessive protrusion in $G$, since $G-B$ may have more components than $H-B$ (they are however not adjacent to $B$).
We will thus consider the property that no component of $G$ has an excessive protrusion.
By this we mean that for each connected component $H$ of $G$, there is no excessive protrusion in $H$.

By exhaustively (at most $\|G\|$ times) executing the algorithm of Lemma~\ref{lem:randomizedContractions} and replacing any obtained protrusion using Lemma~\ref{lem:basicReplacer}, 
we can get rid of all excessive protrusions. 
We formalize this in the following lemma, which will serve as the abstraction of protrusion replacement in the sequel.

\begin{lemma}[Exhaustive Protrusion Replacement]\label{lem:exhaustProtrusions}
There is an algorithm that, given a graph $G$, runs in time $\Oh(\|G\|^2 \log \|G\| \cdot |G|^2)$ 
and computes a graph $G'$ such that $\OPT(G)=\OPT(G')$, $\|G'\|\leq\|G\|$, and no connected component of $G'$ has an excessive protrusion.

Moreover, there exists a solution-lifting algorithm that works as follows: given a subset $F'$ of edges of $G'$ for which $G'-F'$ is $\F$-free, the algorithm runs in time $\Oh(\|G\|^2)$
and outputs a subset $F$ of edges of $G$ such that $|F|\leq |F'|$ and $G-F$ is $\F$-free.
\end{lemma}
\begin{proof}
Inspect every connected component $H$ of $G$, and to each of them apply the algorithm of Lemma~\ref{lem:randomizedContractions}, which runs in time $\Oh(\|G\|\cdot \log \|G\| \cdot |G|^2)$.
This algorithm either concludes that $H$ has no excessive protrusion, or finds some replaceable protrusion $X$ in $H$. Then $X$ is also a replaceable protrusion in $G$,
so by applying Lemma~\ref{lem:basicReplacer} to $X$ we can compute in linear time a new graph $G'$ with $\OPT(G')=\OPT(G)$ and $\|G'\|<\|G\|$. Having found $G'$, we can restart the whole algorithm on $G'$.
Eventually, the algorithm of Lemma~\ref{lem:randomizedContractions} concludes that each component has no excessive protrusions and can hence output $G$.
The solution-lifting algorithm follows by iteratively applying the solution-lifting algorithm of Lemma~\ref{lem:basicReplacer} for all the consecutive replacements performed above.

Since the number of edges strictly decreases in each iteration, the number of iterations is bounded by the number of edges of the original graph $G$.
Therefore, the claimed running time follows.
\end{proof}

\section{Constant-factor approximation}\label{sec:approx}
It would be ideal if just applying the Exhaustive Protrusion Replacement (Lemma~\ref{lem:exhaustProtrusions}) reduced the size of the graph to linear in $\OPT$.
Then, we would already have a linear kernel, and taking all its edges would yield a constant-factor approximation.
Unfortunately, there are graphs with no excessive protrusions, where the size is not bounded linearly in $\OPT$.
To see this, observe that even an arbitrarily large group of parallel edges is not a protrusion, so our current reduction rules will not reduce their multiplicity, even if they amount to 99\% of the graph.
Hence, we need to find a way to discover and account for such groups (we remark here that reducing each to $\Oh(\OPT)$ would be relatively easy, giving a quadratic kernel only, though). 
More generally, the structures that turn out to be problematic are large groups of constant-size $2$-protrusions attached to the same pair of vertices; 
a group of parallel edges is a degenerated case of this structure.
To describe the problematic structures formally, we introduce the notion of a {\em{bouquet}}.

Pruning bouquets and sets of parallel edges to constant size does not give an equivalent graph (because a larger bouquet may always require a larger number of edge deletions).
However, the edge set of the resulting pruned graph intersects some optimal solution of the original graph; this is because for any deleted element, pruning preserves some number of isomorphic elements.
Moreover, since the intersection is a solution for the pruned graph and the pruned graph has no bouquets, we can show that the number of edges of the pruned graph is linear in the size of the intersection.

Pruning thus gives a procedure that finds a subset of edges which intersects an optimal solution and such that the size of the subset is at most a constant factor larger than the size of this intersection.
By iteratively finding such a set and removing it, we obtain a solution that is at most a constant factor larger than the optimum.
In the next section we will leverage the obtained constant-factor approximation to reduce all bouquets at once, thus achieving a linear kernel.

\subsection{Bouquets}
Let us define the following constant (recall that $\MX=\max_{H\in \F}\|H\|$)
$$d_\F := \max\{2b_\F \cdot c_\F+2b_\F,\ 3\MX\}+1$$
We now introduce the notions of \emph{bouquets} and \emph{thetas}. Intuitively, a bouquet is a family of at least $d_\F$ isomorphic $2$-protrusions, while a theta is a set of at least $d_\F$ parallel edges.

\begin{definition}
Consider a graph $G$, a set $U\subseteq V(G)$ and a family of
2-protrusions $\{S_i\}_{i\in I}$ such that for each $i\in I$:
\begin{itemize}
\item $N(S_i)=U$ (implying $|U|\leq 2$); 
\item $G[S_i]$ is connected; and 
\item $G[U \cup S_i]$ is isomorphic to $G[U \cup S_j]$ for all $i,j \in I$,\\
with an isomorphism that maps each vertex of $U$ to itself.
\end{itemize}
We call such a family a \emph{bouquet attached to $U$} if it is maximal under inclusion (i.e. there is no proper superfamily which is also a bouquet) and has at least $d_\F$ elements.
The edge set of the bouquet is the set of all edges incident to some $S_i$.
\end{definition}
\begin{definition}
For two vertices $u,v\in V(G)$, a \emph{theta attached to $\{u,v\}$} is a set of edges between $u$ and $v$ that is maximal under inclusion and has at least $d_\F$ elements.
\end{definition}

The constant $d_\F$ is chosen so that a protrusion containing a set to which a bouquet (or theta) is attached is large enough to be excluded as an excessive protrusion, 
and so that any immersion of a graph of $\F$ cannot simultaneously intersect all elements of a bouquet. 
Indeed, in any immersion of some $H\in \F$ in a graph $G$, the image of an edge of $H$ is a path in $G$, which visits every vertex of the bouquet's attachment at most once, 
and hence intersects at most three elements of the bouquet.
Thus in total, the immersion model intersects at most $\max_{H\in \F} 3\|H\|=3\MX$ elements of the bouquet or theta, which is less than $d_\F$.

\medskip
We now show that the number of edges of a graph with no excessive protrusions, no bouquets and no thetas is linearly bounded in the optimum solution size, which formalizes the intuition that these structures are the only obstacles preventing the graph from being a linear kernel.

The following well-known notion and lemma are useful for proving such bounds.
For a rooted forest $T$ and a set $M\subseteq V(T)$, the {\em{least common ancestor closure}} ({\em{lca-closure}}) of $M$ is the set $\LCA(M)\subseteq V(T)$ obtained from $M$ by 
repeatedly adding to it the least common ancestor of every pair of nodes in the set (unless the nodes are in different connected components of the forest $T$).

\begin{lemma}[\!\!\cite{FominLMS12}]\label{lem:lca}
	Let $T$ be a rooted forest and $M\subseteq V(T)$.
	Then $|\LCA(M)|\leq 2 |M|$ and every connected component $C$ of $T - \LCA(M)$ has at most two neighbors in $T$.
\end{lemma}

\begin{lemma}\label{lem:nobouquetsLinear}
	Let $G$ be a connected graph with no excessive protrusions, no bouquets and no thetas.
	Then either $G$ is $\F$-free, or $\|G\| \leq c \cdot \OPT(G)$ for some constant $c$ depending on $\F$ only.
\end{lemma}
\begin{proof}
	Denote $k:=\OPT(G)$ and suppose $G$ is not $\F$-free, that is, $k \geq 1$.
	Let $F\subseteq E(G)$ be a set of $k$ edges such that $G - F$ is $\F$-free.
	By Corollary~\ref{corol:neat}, $G - F$ has a neat tree-cut decomposition $(T,\mathcal{X} = \{X_t, : t \in V(T)\})$ with $\width'(T,\mathcal{X})\leq b_\F$.
	Let $M\subseteq V(T)$ be the lca-closure of the set of the nodes of $T$ that correspond to the bags which contain some vertices incident to edges of $F$.
	Since $|F|\leq k$ there are at most $2k$ such bags and, by Lemma~\ref{lem:lca}, we have that $|M| \leq 4k$ and that every connected component of $T-M$ has at most two neighbors in $T$;
	see Figure~\ref{fig:kernel}.
	Recall that by $X_{T'}$ we denote the union of bags at the nodes of a subtree $T'$ of~$T$.
	
	\begin{claim}\label{cl:rest-protr}
	For any connected component $T'$ of $T-M$,\ $\|G[X_{T'}]\| \leq 2b_\F \cdot c_\F.$
	\end{claim}
	\begin{proof}
	Suppose to the contrary that $\|G[X_{T'}]\| > 2b_\F \cdot c_\F$.
	We verify that then $X_{T'}$ is an excessive protrusion.
	Indeed, $X_{T'}$ has no vertices incident to $F$, so $G[X_{T'}]$ is $\F$-free.
	Moreover, $T'$ has at most two neighbors in $T$, so $|\delta_G(X_{T'})|\leq 2b_\F$ and $T - V(T')$ has at most two components adjacent to $T'$.
	By the properties of neat decompositions, 
	the unions of bags of these two components of $T - V(T')$ induce at most two connected components in $G-F$; in other words, $G-F-X_{T'}$ has at most two connected components adjacent to $X_{T'}$, say $C_1, C_2$.
	Then $N_{G-F}(X_{T'}) \subseteq C_1 \cup C_2$ and since $X_{T'}$ has no vertices incident to $F$, also  $N_G(X_{T'}) \subseteq C_1 \cup C_2$.
	Since $G$ is connected, this implies $G-X_{T'}$ has at most two connected components (because every vertex of $G-X_{T'}$ has a path connecting it to $X_{T'}$ in $G$, which must visit $N_G(X_{T'})$).
	This shows that  $X_{T'}$ is an excessive protrusion, contradicting assumptions.
	\cqed\end{proof}
	
	\noindent We set $c':=2b_\F \cdot c_\F$.
	
	Each connected component of $T - M$ has exactly one or exactly two neighbors in $M$;
	it cannot have zero neighbors in $M$, as it would then induce a component in $G-F$ with no vertices incident to $F$, contradicting that $G$ is connected and not $\F$-free.
	Observe that the number of components that have exactly two neighbors in $M$ is at most $|M|-1$,
	because replacing each such component with an edge connecting its neighbors yields a forest with vertex set $M$.
	It remains to bound the number of components in $T - M$ with exactly one neighbor in $M$.
	
	Suppose $T_1,\ldots,T_p$ are those connected components of $T - M$ for which the neighborhood in $T$ is exactly $t$, for some $t\in M$.
	Again, by $X_{T_i}$ we denote the union of the bags at the nodes of $T_i$.
	By Corollary~\ref{corol:neatAdhesions}, at most $2b_\F+1$ of them are not connected with neat adhesions to $t$, so assume w.l.o.g. that $T_1,\dots,T_{p-2b_\F-1}$ are.
	That is, $N_G(X_{T_i})$ is a non-empty subset of $X_t$ of size at most two, for all $i=1,\dots,p-2b_\F-1$.
	Since there are at most $b_\F^2$ such subsets of size at most $2$,
	at least $\frac{p-2b_\F-1}{b_\F^2} \geq p/b_\F^2 - 3$ of the sets $X_{T_i}$ have the same neighborhood $U$ in $G$.
	By the neatness of the decomposition, 
	each $G[X_{T_i}]$ is connected and moreover $\|G[X_{T_i}]\| \leq c'$  by Claim~\ref{cl:rest-protr}. This implies in particular that $|X_{T_i} \cup U| \leq c'  +3$.
	This means that there are at most $(c'+4)^{2c'}$ possible isomorphism types for $G[X_{T_i} \cup U]$.
	If there were at least $d_\F$ components with the same isomorphism type, they would form a bouquet.
	Hence $p/b_\F^2 - 3 \leq (c'+4)^{2c'} \cdot d_\F$,
	meaning that
\[
p \leq ((c'+4)^{2c'} \cdot d_\F + 3) \cdot b_\F^2.
\]
We define $c'':=((c'+4)^{2c'} \cdot d_\F + 3) \cdot b_\F^2$.
	
	Therefore, $T-M$ is partitioned into at most $|M|-1 + c'' \cdot |M| \leq (c'' + 1)\cdot |M|$ connected components.
	By Claim~\ref{cl:rest-protr}, for each of these components, the vertices contained in its bags induce a subgraph with at most $c'$ edges.
	In addition to these edges, the edge set of $G$ contains only: 
	\begin{itemize}
	\item $k$ edges of the deletion set $F$;
	\item $d_\F \cdot b_\F^2$ edges in $G[X_t]$ for each $t\in M$ ($G[X_t]$ has at most $b_\F$ vertices and every pair has less than $d_\F$ edges in between, as $G$ has no thetas); and
	\item up to $\left((c'' +1)\cdot|M|+|M|\right) \cdot  b_\F=(c''+2)\cdot |M|\cdot b_\F$ edges between parts of the partition of $V(T)$ given by individual elements of $M$ 
	      and connected components of $T - M$ (each edge of $T$ between different parts yields at most $b_\F$ edges).
	\end{itemize}
	Since $|M|\leq 4k$, we infer that the number of edges in $G$ is at most
\[
4(c'' + 1)k \cdot c'\ +\ k\ +\ 4k\cdot d_\F \cdot b_\F^2\ +\ 4(c'' +2)k \cdot  b_\F \ =\ c'''\cdot k,
\]
for a constant $c''':=4c'(c'' + 1)+1+4 d_\F \cdot b_\F^2+4b_\F(c'' +2)$.
\end{proof}

\begin{figure}[t]
	\centering
	{\input{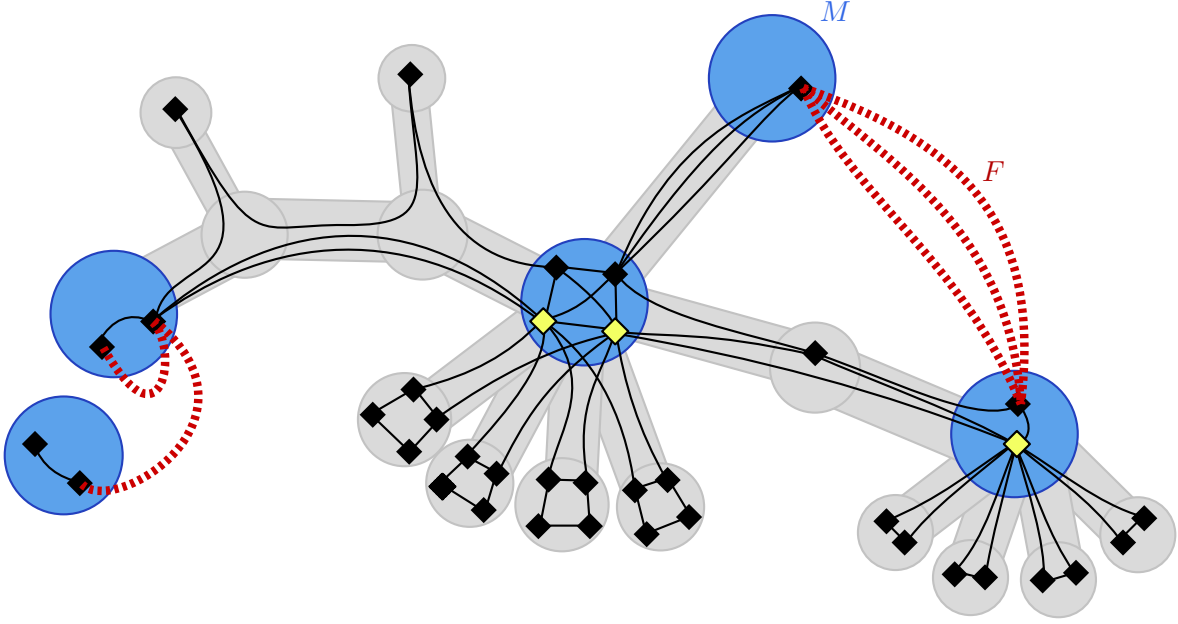}}
	\caption{
		A graph $G$ with a solution $F$ (five red dotted edges), in a tree-cut decomposition of $G-F$.
		The set $M$ used in Lemma~\ref{lem:nobouquetsLinear} is highlighted in blue: three bags incident to $F$, a fourth central one in their lca-closure. 
		Remaining components turn out to have sizes bounded by a constant (light gray). Two bouquets are present, their attachments visible as light yellow vertices.}
	\label{fig:kernel}
\end{figure}

\subsection{Finding a constant-factor approximation piece by piece}
To handle bouquets and thetas algorithmically, we first show that they are disjoint, as otherwise they would constitute a large protrusion.
In this subsection, we frequently use the observation that, 
in a connected graph, a 2-protrusion with more than $2b_\F \cdot c_\F$ edges is an excessive protrusion. 
Indeed, if $X$ is a $2$-protrusion in a connected graph $G$, then it is always the case that $G-X$ has at most two connected components, due to $|\delta(X)|\leq 2$.

\begin{lemma}\label{lem:disjoint}
	Let $G$ be a connected graph with no excessive protrusions.
	Then every two bouquets and/or thetas in $G$ have disjoint edge sets.
	Furthermore, if a bouquet or theta is attached to $U\subseteq V(G)$, then $U$ is disjoint with all elements of any bouquet.
\end{lemma}
\begin{proof}
	Suppose a bouquet or theta $\mathcal{X}$ is attached to $U_X$ and a bouquet $\mathcal{Y} = \{Y_j\}_{j\in J}$ is attached to $U_Y$.
	Observe that either $|U_X| = 1$ or the two vertices of $U_X$ are joined by at least $d_\F > 2$ edge-disjoint paths in $G$;
	so in any case all of $U_X$ is on one side of the cut $\delta(Y_j)$, for each $j\in J$.
	If $U_X \subseteq Y_j$ for some $j$, then all but possibly two elements of $\mathcal{X}$ would be contained in the side $Y_j$ of this cut.
	Hence $Y_j$ would be a 2-protrusion with $\|G[Y_j]\| \geq d_\F-2 > 2b_\F \cdot c_\F$ and thus an excessive protrusion, contradicting assumptions.
	We infer that $U_X$ is disjoint with the elements of any bouquet, which concludes the second part of the claim.
	
	To show the first part of the claim, first notice that two thetas cannot have intersecting edge sets, as they are maximal sets of parallel edges.
	Secondly, if a theta contained an edge from the edge set of a bouquet $\mathcal{Y} = \{Y_j\}_{j\in J}$, then either it would be an edge of $\delta(Y_j)$ for some $j\in J$, contradicting $|\delta(Y_j)|\leq 2$;
	or it would be an edge of $G[Y_j]$, again implying that $\|G[Y_j]\| > 2b_\F \cdot c_\F$ and that thus $Y_j$ is an excessive protrusion, contradicting assumptions.
	
	Finally, consider the case $\mathcal{X}=\{X_i\}_{i\in I}$ is a bouquet attached to $U_X$ and  $\mathcal{Y} = \{Y_j\}_{j\in J}$ is a bouquet attached to $U_Y$.	
	Let $X=\bigcup_{i\in I} X_i$ and $Y=\bigcup_{j\in J} Y_j$. 
	We already showed that $U_X\cap Y=\emptyset$ and symmetrically $U_Y\cap X=\emptyset$.
	Therefore $U_X, U_Y \subseteq V(G)\setminus (X\cup Y)$.
	We infer that the edge-sets of the two bouquets can intersect only if the sets $X$ and $Y$ intersect.
	Hence $X_i \cap Y_j \neq \emptyset$ for some $i\in I,j \in J$.
	If there was an edge between some $u\in X_i \cap Y_j$ and some $v\in Y_j \setminus X_i$, then $v \in N(X_i) = U_X$ and $v\in Y_j$, contradicting $U_X \cap Y = \emptyset$.
	Hence there is no edge between $X_i \cap Y_j$ and $Y_j \setminus X_i$.
	Since $Y_j$ is connected and $X_i\cap Y_j$ is non-empty, this implies $Y_j \subseteq X_i$.
	A symmetric reasoning yields that $X_i \subseteq Y_j$, so $X_i = Y_j$.
	But then $U_X=U_Y$ and all the elements of both bouquets are pairwise isomorphic with an isomorphism that fixes $U_X=U_Y$.
	We infer that the union of the two bouquets is also a bouquet.
	Hence, by maximality of bouquets, we conclude that $\{X_i\}_{i\in I}$ and $\{Y_j\}_{j\in J}$ are in fact the same bouquet.
\end{proof}

We now proceed to formalizing the procedure of pruning all bouquets and thetas to constant size. 
We show that the remaining edge set, named $\Delta$, has the property that its removal from $G$ strictly decreases $\OPT(G)$ and its size is linear in that decrement.

\begin{lemma}\label{lem:findEdge}
	Given a connected graph $G$ with no excessive protrusion that is not $\F$-free, one can find in in time $\Oh(|G|^3)$ a set $\Delta \subseteq E(G)$ such that (for some $c$ depending on $\F$ only):
	$$\OPT(G-\Delta)<\OPT(G)\text{\quad and \quad}|\Delta|\leq c \cdot (\OPT(G) - \OPT(G-\Delta)).$$
\end{lemma}
\begin{proof}
	The algorithm finds all bouquets and thetas and deletes all but $d_\F-1$ elements from each.
	More precisely, the algorithm first deletes all but $d_\F-1$ edges from each theta in $G$, resulting in a subgraph $G'$.
	Note that since $d_\F-1>2b_\F$, this reduction of thetas cannot introduce excessive protrusions in the graph and hence remaining bouquets are disjoint in the sense of Lemma~\ref{lem:disjoint}. 
	The algorithm then sets $V' := V(G') = V(G)$, finds all bouquets in $G'$ and deletes from $V'$ all vertices of all but $d_\F-1$ elements of each bouquet.
	The algorithm then outputs $\Delta := E(G'[V'])$.
	
	Bouquets in $G'$ can be found by checking all possible attachments $U$ of size at most $2$ and all components of $G-U$ containing at most $2b_\F \cdot c_\F$ edges 
	(2-protrusions cannot have more edges, as they would form excessive protrusions otherwise).
	There are $\Oh(|G|^2)$ possible attachments $U$ and checking all components of $G-U$ for any $U$ takes time $\Oh(|G|)$, hence the running time follows.
	
	To prove that $\Delta = E(G'[V'])$ has the claimed properties, let us first show that $G'[V']$ is not $\F$-free; we use the following slightly more general statement later.
	\begin{claim}\label{cl:notFfree}
		Let $S\subseteq V(G)$.
		If $G'[V' \cap S]$ is $\F$-free, then so is $G'[S]$.
		If $G'[S]$ is $\F$-free, then so is $G[S]$.	
		In particular, $G'[V']$ is not $\F$-free.
	\end{claim}
	\begin{proof}
		Suppose $G[S]$ (or $G'[S]$) is not $\F$-free.
		Then there is an immersion model of a graph from $\F$ in $G[S]$ (or $G'[S]$).
		Since such a model intersects at most $d_\F-1$ elements of any theta or bouquet,
		we can find an immersion model that intersects only the elements that were not deleted from $G$ when constructing $G'$, nor from $V'$ when constructing $G'[V']$.
		This means $G'[S]$ and $G'[V'\cap S]$ also contain an immersion of a graph in $\F$, that is, they are not $\F$-free.
	\cqed\end{proof} 
	
	Consider now an optimal solution $F\subseteq E(G)$ for $G$.
	Then $F \cap \Delta$ is a solution (not necessarily optimal) for the subgraph $G'[V']$,
	 meaning $F\cap \Delta$ is non-empty (as $G'[V']$ is not $\F$-free) and $\OPT(G'[V'])\leq |F \cap \Delta|$.
		
	Observe that since $\OPT(G)=|F|$, we have $\OPT(G)-\OPT(G-F \cap \Delta) = |F \cap \Delta|$.
	Hence
	$$\OPT(G)-\OPT(G-\Delta) \geq |F \cap \Delta|\mbox{,\quad and in particular, \quad}\OPT(G-\Delta) < \OPT(G).$$	
	
	We show in the three claims below that $G'[V']$ is connected and has no excessive protrusions, no bouquets and no thetas.
	Therefore, by Lemma~\ref{lem:nobouquetsLinear}, $G'[V']$ has at most $c \cdot \OPT(G'[V'])$ edges, for some constant $c$ depending on $\F$ only.
	Using the above inequalities, we reach the desired conclusion:
	$$|\Delta| = |E(G'[V'])| \leq c \cdot \OPT(G'[V']) \leq c \cdot |F \cap \Delta| \leq c \cdot (\OPT(G)-\OPT(G-\Delta)).$$	
	
	It remains to show that $G'[V']$ is indeed connected, has no excessive protrusions, no bouquets and no thetas.
	Clearly $G'$ has no theta and thus $G'[V']$ has no theta either.
	Any edge or path deleted in the construction can be replaced with one that was not deleted, hence $G'[V']$ is connected.

\begin{claim}
	$G'$ has no excessive protrusions.
\end{claim}
\begin{proof}
	Suppose $S\subseteq V(G')=V(G)$ is an excessive protrusion in $G'$.	
	If there was a theta in $G$ attached to a set $U$ with one vertex in $S$ and the other outside of $S$, then in $G'$ there would still be least $d_\F-1$ edges connecting $S$ and $V(G')\setminus S$, contradicting that $|\delta_{G'}(S)| \leq 2b_\F < d_\F - 1$.
	Hence there is no such theta and, in particular, $\delta_{G'}(S) = \delta_G(S)$.
	
	Since $S$ is not an excessive protrusion in $G$, 
	there must be a theta in $G'$ with an attachment $U$ contained in $S$.
	Thus $G[S]$ has at least $d_\F$ edges, which is more than $2b_\F \cdot c_\F$.
	As $S$ is an excessive protrusion in $G'$, $G'-S$ has at most two connected components and thus so does $G-S$.
	Note that since $S$ is a $2b_\F$-protrusion in $G'$, we have $|\delta_G(S)| = |\delta_{G'}(S)| \leq 2b_\F$ and since $G'[S]$ is $\F$-free, so is $G[S]$ (Claim~\ref{cl:notFfree}).
	Therefore $S$ is an excessive protrusion in $G$, a contradiction.
\cqed\end{proof}	
	
\begin{claim}\label{claim:noProt}
	$G'[V']$ has no excessive protrusions.
\end{claim}
\begin{proof}
	Suppose $S$ is an excessive protrusion in $G'[V']$.
	Since it is not an excessive protrusion in $G'$, it must be the case that $S$ had more incident edges in $G'$ than in $G'[V']$, i.e. $\delta_{G'}(S)\supsetneq \delta_{G'[V']}(S)$. 
	By the construction of $G'$, we infer that there is a bouquet $\{Y_j\}_{j\in J}$ in $G'$ attached to some $U_Y$ and a $j\in J$ such that $Y_j$ was removed when constructing $V'$ and was adjacent to $S$.
	Since $N(Y_j)=U_Y$, we have that $U_Y\cap S \neq \emptyset$.
	
	We claim that $U_Y\subseteq S$. If $U_Y$ has one vertex this is clear. Otherwise, if $U_Y$ has two vertices, then at least $d_\F -1$ elements of the bouquet connect them in $G'[V']$, yielding a family of more than $2b_\F$ edge-disjoint paths between them.
	Thus $U_Y$ lies entirely on one side of the cut $\delta_{G'[V']}(S)$, because $|\delta_{G'[V']}(S)|\leq 2b_\F$. 
	Since $U_Y\cap S \neq \emptyset$, we conclude that $U_Y \subseteq S$.
	
	Similarly, we deduce that every other bouquet in $G'$ has an attachment fully contained in either $S$ or in $V' \setminus S$.
	Recall that $V'$ is obtained from $V(G')$ by removing vertices of some elements from each bouquet.
	It follows that $\delta_{G'[V']}(S)$ is a cut in $G'$ too.
	More precisely, let $S^*\subseteq V(G')$ be the set consisting of $S$ and those elements of bouquets deleted when constructing $V'$ that were attached to vertices in $S$. 
	Then $\delta_{G'}(S^*) = \delta_{G'[V']}(S)$.
	
	Note that as $S$ is a $2b_\F$-protrusion in $G'[V']$, we have $|\delta_{G'}(S^*)| = |\delta_{G'[V']}(S)| \leq 2b_\F$ and since $G'[S]=G'[V' \cap S^*]$ is $\F$-free, so is $G'[S^*]$ (Claim~\ref{cl:notFfree}).
	Thus $S^*$ is a $2b_\F$-protrusion in $G'$.
	
	Since $U_Y\subseteq S \subseteq S^*$, $S^*$ must contain all elements of $\{Y_j\}_{j  \in J}$ except for at most $|\delta_{G'}(S^*)|\leq 2b_\F$. Consequently,
	$$\|G'[S^*]\|\geq d_\F -2b_\F > 2 b_\F \cdot c_\F.$$
	As $S$ is an excessive protrusion in $G'[V']$, we have that $G'[V']-S$ has at most two connected components.
	The graph $G'-S^*$ can be obtained from $G'[V']-S$ by reintroducing elements of bouquets (which induced connected subgraphs) with attachments in $V'\setminus S$, hence $G-S^*$ also has at most two components.
	Therefore, $S^*$ is an excessive protrusion in $G'$, a contradiction.
\cqed\end{proof}

\begin{claim}
	$G'[V']$  has no bouquet.
\end{claim}
\begin{proof}
	Suppose $G'[V']$ has a bouquet $\{X_i\}_{i\in I}$ attached to some $U_X\subseteq V'$.
	Since it was not removed when constructing $V'$, it was not a bouquet in $G'$, so it must be that $\delta_{G'[V']}(X_i) \subsetneq \delta_{G'}(X_i)$ for some $i\in I$.
	By the construction of $V'$, there must be some bouquet $\{Y_j\}_{j\in J}$ attached to some $U_Y$ in $G$ with some $Y_j, j\in J$ adjacent to $X_i$.
	Again, $U_Y$ either consists of one vertex, or of two vertices that are still connected in $G'[V']$ by $d_\F-1>2$ edge-disjoint paths (namely paths contained in the elements of $\{Y_j\}_{j\in J}$ that did not get deleted).
	Hence all of $U_Y$ lies on the same side of the cut $\delta(X_i)$ in $G'[V']$.
	Since $Y_j$ is adjacent to $X_i$, $N_{G'}(Y_j)=U_Y$ intersects $X_i$ and thus $U_Y \subseteq X_i$.
	As $|\delta_{G'[V']}(X_i)|\leq 2$, we have that $\delta_{G'[V']}(X_i)$ can intersect at most two of the $d_\F-1$ elements of the bouquet $\{Y_j\}_{j\in J}$ that survive in $G'[V']$.
	All the other elements of this bouquet must lie on the same side of the cut $\delta(X_i)$ as $U_Y$, that is, they must be contained in $X_i$.	
	So in fact $X_i$ is a 2-protrusion in $G'[V']$ containing at least $d_\F-3 > 2b_\F \cdot c_\F$ edges and thus an excessive protrusion, a contradiction.
	\looseness=-1
\cqed\end{proof}

	With the above claims, we conclude that $G'[V']$ has no excessive protrusions, no bouquets and no thetas. Therefore, it satisfies the conditions of Lemma~\ref{lem:nobouquetsLinear}.
\end{proof}

We extend Lemma~\ref{lem:findEdge} to disconnected graphs by simply considering each connected component separately.

\begin{corollary}\label{cor:findEdge}
	Suppose we are given a graph $G$ that is not $\F$-free and in which every connected component does not have any excessive protrusion. 
	Then one can find in time $\Oh(|G|^3)$ a set $\Delta \subseteq E(G)$ such that (for some constant $c$ depending on $\F$ only):
	$$\OPT(G-\Delta)<\OPT(G)\text{\quad and \quad}|\Delta|\leq c \cdot (\OPT(G) - \OPT(G-\Delta)).$$
\end{corollary}
\begin{proof}
	Let $C_1, \dots, C_r$ be the connected components of $G$ and let without loss of generality $C_1,\dots,C_{r'}$ be those that are not $\F$-free, for some $r'\leq r \in \mathbb{N}$.
	Since $G$ is not $\F$-free and all graphs of $\F$ are connected, we have $r'\geq 1$.
	Since every component $C_i$ has no excessive protrusion, for each $1\leq i \leq r'$ we compute a set $\Delta_i \subseteq E(C_i)$ by invoking Lemma~\ref{lem:findEdge} on $C_i$.
	Then, for some constant $c$ depending on $\F$ only, we have that:
	$$\OPT(C_i-\Delta_i)<\OPT(C_i)\text{\quad and \quad}|\Delta_i|\leq c \cdot (\OPT(C_i) - \OPT(C_i-\Delta_i))$$	
	Let $\Delta = \bigcup_{i=1}^{r'} \Delta_i$.
	Since all the graph in $\F$ are connected, $\OPT(G) = \sum_{i=1}^r \OPT(C_i)$.
	Components $C_i$ that are $\F$-free have $\OPT(C_i)= 0$, thus
	$$\displaystyle\OPT(G) = \sum_{i=1}^{r'} \OPT(C_i)\text{\quad and similarly \quad}
	\OPT(G-\Delta) = \sum_{i=1}^{r'} \OPT(C_i - \Delta_i)$$
	Since $r'\geq 1$, we have $\OPT(C_1-\Delta_1)<\OPT(C_1)$ and thus $\OPT(G-\Delta) < \OPT(G)$.
	Finally
	$$ |\Delta| = \sum_{i=1}^{r'} |\Delta_i| \leq c \cdot (\OPT(G) - \OPT(G-\Delta)).$$	
	Therefore, the algorithm can in $\Oh(\sum_{i=1}^{r'} |C_i|^3) \leq \Oh(|G|^3)$ time output $\Delta$ as a result.
\end{proof}

To get a constant-factor approximation algorithm, we invoke the above corollary iteratively.
Intuitively, we maintain a set of edges $F$, initially empty, and invoke the corollary on $G-F$ to find a set $\Delta$ such that adding it to $F$ decreases $\OPT(G-F)$, while increasing $|F|$ by only a constant factor more.
We then run the algorithm of Lemma~\ref{lem:exhaustProtrusions} to remove excessive protrusions from $G-F$, reducing in a sense those parts of the graph where no more edges need to be deleted.
Eventually, we reach $\OPT(G-F)=0$, meaning $F$ is a solution of size linear in $\OPT(G)$.
The proof is straightforward, but requires reconstructing at every step a solution to the original graph given to Lemma~\ref{lem:exhaustProtrusions}.

\begin{theorem}[Theorem~\ref{thm:main-apx}, reformulated]\label{thm:approx}
	There is an algorithm running in time $\Oh(\|G\|^3 \log \|G\| \cdot |G|^3)$
	that given a graph $G$, outputs a set $F\subseteq E(G)$ of size at most $c_{\mathrm{apx}} \cdot \OPT(G)$ such that $G-F$ is $\F$-free, for some constant $c_{\mathrm{apx}}$ depending on $\F$ only.
\end{theorem}
\begin{proof}
	Let $G_0 = G$ and $\Delta_0 = \emptyset$.	
	The algorithm computes a sequence of graphs $G_i$ with $\|G_i\|\leq \|G\|$ and sets $\Delta_i \subseteq E(G_i)$ as follows.
	
	For $i \geq 0$, $G_{i+1}$ is computed from $G_i - \Delta_i$ by invoking Lemma~\ref{lem:exhaustProtrusions} on this graph.
	That is, in time $\Oh(\|G\|^2 \log \|G\| \cdot |G|^3)$ we compute a graph $G_{i+1}$ in which no connected component contains any excessive protrusions and moreover
	$$\|G_{i+1}\| \leq \|G_i - \Delta_i\| \leq \|G\|\qquad\textrm{and}\qquad \OPT(G_{i+1})=\OPT(G_i - \Delta_i).$$
	
	For $i \geq 1$, provided $G_i$ is not $\F$-free, $\Delta_i$ is computed from $G_i$ by invoking Corollary~\ref{cor:findEdge}.
	That is, in time $\Oh(|G|^3)$ we find a set $\Delta_i\subseteq E(G_i)$ such that
	$$\OPT(G_i-\Delta)<\OPT(G_i)\text{\quad and \quad}|\Delta_i|\leq c \cdot (\OPT(G_i) - \OPT(G_i-\Delta_i)).$$
	Here, $c$ is the constant given by Corollary~\ref{cor:findEdge}.
	
	Eventually, since $\OPT(G_{i+1})<\OPT(G_i)$, there is an $1\leq r \leq \OPT(G)$ such that $G_r$ is $\F$-free.
	We reconstruct a sequence of solutions $F_i \subseteq E(G_i)$ for $G_i$ as follows.
	Clearly $F_r := \emptyset$ is a solution for $G_r$.
	If $F_{i+1}$ is a solution for $G_{i+1}$, then
	using the solution-lifting algorithm of Lemma~\ref{lem:exhaustProtrusions},
	a solution $F'_i$ for $G_i - \Delta_i$ can be constructed in time $\Oh(\|G\|^2)$ such that $|F'_i| \leq |F_{i+1}|$.
	Then $F_i := F'_i \cup \Delta_i$ is a solution for $G_i$.
	This way, we reconstruct a solution $F_0$ for $G_0=G$ in at most $\OPT(G)$ iterations, that is, in total time $\Oh(\|G\|^2 \cdot \OPT(G))$.
	Constructing the graphs $G_i$ took $\Oh(\|G\|^2 \log \|G\| \cdot |G|^3 \cdot \OPT(G))$ total time, so since $\OPT(G)\leq \|G\|$, the time bound follows.
	
	To bound the size of the solution, we show inductively that $|F_i| \leq c \cdot \OPT(G_i)$ for $i=r,\dots,0$.
	Clearly this hold for $i=r$.
	If it holds for $i+1$, then it holds for $i$, because:
	$$|F_i| \leq |F'_i| + |\Delta_i| \leq |F_{i+1}| + |\Delta_i| \leq c \cdot \OPT(G_{i+1}) + c \cdot ( \OPT(G_{i}) -  \OPT(G_{i+1})) = c \cdot \OPT(G_i).\qed$$
	\phantom\qedhere
\end{proof}
\vspace{-9mm}

\section{Linear kernel}\label{sec:kernel}

In the previous section we have already observed (Lemma~\ref{lem:nobouquetsLinear}) that the only structures in the graph that prevent it from being a linear kernel
are excessive protrusions, bouquets, and thetas. Using the Exhaustive Protrusion Replacement (Lemma~\ref{lem:exhaustProtrusions}) we can 
get rid of excessive protrusions, but bouquets and thetas can still be present in the graph.

It would be ideal if we could reduce the size of every bouquet or theta to a constant, but unfortunately we are so far unable to do this.
Instead, we employ the following strategy based on the idea of {\sl amortization}.
First, we reduce all excessive protrusions using Lemma~\ref{lem:exhaustProtrusions}.
Second, using Theorem~\ref{thm:approx}, we compute an approximate solution $F$ that is larger than the optimum only by a constant multiplicative factor.
Then, we investigate every bouquet in the graph and we estimate the number of edges of $F$ that, in some sense, ``affect'' the bouquet.
It can be then shown that the size of the bouquet can be reduced to linear in terms of the number of edges that affect it.
Thus, after performing this reduction there still might be large bouquets in the graph, but only because a large number of edges of $F$ affect them.
However, every edge of $F$ will affect at most a constant number of bouquets, so the total size of the bouquets will amortize to linear in terms of $|F|$, hence also linear in terms of $\OPT$.
The same amortization reasoning also enables us to bound the total sum of sizes of thetas.

We first show that if we know a local solution that isolates a bouquet into an $\F$-free part, then this bouquet can be proportionally bounded without changing $\OPT(G)$.

\begin{lemma}\label{lem:bouquetBound}
	Let $\{X_i\}_{i\in I}$ be a bouquet attached to $U$ in $G$. 
	Suppose $\Delta\subseteq E(G)$ is such that all the connected components of $G-\Delta$ that intersect $U \cup \bigcup_{i\in I} X_i$ are $\F$-free.
	Then $\OPT(G)=\OPT(G')$, where $G'$ is obtained from $G$ by removing vertices of all except $d_\F+|\Delta|$ elements of the bouquet.
\end{lemma}
\begin{proof}
	Suppose that, to the contrary, $G'$ admits an edge subset $F\subseteq E(G')$ of size $k$ such that $G'-F$ is $\F$-free, but $G$ does not.
	Let $C\subseteq V(G)$ be the union of the vertex sets of those connected components of $G-\Delta$ that contain some vertices of $U \cup \bigcup_{i\in I} X_i$. 
	By assumption, $G[C]-\Delta$ is $\F$-free.
	Note that $\delta(C) \subseteq \Delta$.
	
	Let $E_C$ be the set of all the edges incident to vertices of $C$ in $G$.
	We claim that $|\Delta| > |F\cap E_C|$.
	To show this, define $F' := (F \setminus E_C) \cup \Delta$.
	Observe that $G-F'$ is $\F$-free: as all graphs in $\F$ are connected, an immersion model of one of them in $G-F' \subseteq G-\delta(C)$ would either be contained in $C$ or disjoint from it.
	The first case would contradict the assumption that $G[C]-F' \subseteq G[C]-\Delta$ is $\F$-free.
	In the second case, the immersion model would be contained in $V(G)\setminus C$, which is equal to $V(G')\setminus C$, because $C$ contains all the vertices of the bouquet.
	That is, it would be contained in in $G[V(G)\setminus C] - F' \subseteq G'[V(G')\setminus C] - (F\setminus E_C) = G'[V(G')\setminus C] - F$, which would contradict the assumption that $G'-F$ is $\F$-free. 
	By our supposition that $G$ does not admit a solution of size $k$, we infer that 
	$|F'|>k\geq |F|$, 
	and hence $|\Delta| > |F\cap E_C|$, as claimed.
	
	Since $G-F$ cannot be $\F$-free, by our assumption that $G$ has no solution of size $k$, it contains an immersion model of some graph in $\F$.
	As argued after the definition of a bouquet, this immersion model can intersect at most $d_\F$ elements of the bouquet $\{X_i\}_{i\in I}$.
	Since $G'$ still has $d_\F+|\Delta|$ isomorphic elements of this bouquet,
	$G'-F$ has at least $d_\F+|\Delta| - |F\cap E_C|> d_\F$ isomorphic elements that are not intersected by $F$.
	Therefore, even if the immersion model in $G-F$ intersected any elements removed from $G'$,
	they can be replaced by not intersected elements that remained unchanged in $G'-F$, thus yielding an immersion model of the same graph in $G'-F$.
	This means that $G'-F$ is not $\F$-free, a contradiction.
\end{proof}

The same reasoning can also be applied to limit the sizes of thetas. The proof is exactly the same  and hence we leave it to the reader.

\begin{lemma}\label{lem:thetaBound}
	Let $\{e_i\}_{i\in I}$ be a theta attached to $\{u,v\}=U$ in $G$.
	Suppose $\Delta\subseteq E(G)$ is such that all the connected components of $G-\Delta$ that contain some vertex of $U$ are $\F$-free.
	Then $\OPT(G)=\OPT(G')$, where $G'$ is obtained from $G$ by removing all edges of $\{e_i\}_{i\in I}$ except for $d_\F + |\Delta|$.
\end{lemma}

The above lemmas allow us to reduce bouquets and sets of parallel edges effectively, given a local part of an approximate solution.
By appropriately amortizing bounds with the total size of the approximate solution, we finally get a linear bound on an irreducible equivalent instance. 

\begin{lemma}\label{lem:reduceFromApprox}
	Let $G$ be a connected graph with no excessive protrusions and let $F\subseteq E(G)$ be such that $G-F$ is $\F$-free.
	Then either $\|G\| \leq c \cdot |F|$ for some constant $c$ depending on $\F$ only, or given $G$ and $F$, 
	one can compute in time $\Oh(\|G\| \cdot |G|^2)$ a subgraph $G'$ of $G$ such that $\OPT(G)=\OPT(G')$ and $\|G'\|<\|G\|$.
\end{lemma}
\begin{proof}
We begin as in the proof of Lemma~\ref{lem:nobouquetsLinear}, except that given $F$ we can now do the same effectively.
That is, using Corollary~\ref{corol:neat}, we compute a neat tree-cut decomposition $\T = (T, \mathcal{X})$ of $G-F$ with $\width'(\T)\leq b_\F$ in time $\Oh(\|G\|\cdot |G|^2)$.
For a node $t$ of $T$, by $X_t$ we denote the bag at~$t$.

Let $M\subseteq V(T)$ be the lca-closure of the set of bags containing a vertex incident to $F$; $M$ can be easily computed in linear time.
By Lemma~\ref{lem:lca}, $|M|\leq 4|F|$ and every connected component of $T-M$ has at most two neighbors in $T$.
Then Claim~\ref{cl:rest-protr} from the proof of Lemma~\ref{lem:nobouquetsLinear} can be argued exactly in the same manner.
We recall it for convenience and refer the reader to the proof of Claim~\ref{cl:rest-protr} for the argumentation.

\begin{claim}[Claim~\ref{cl:rest-protr}, restated] \label{cl:rest-protr2}
Take any connected component $T'$ of $T-M$ 
and let $X_{T'}$ be the union of the bags at the nodes of $T'$. Then
\[
\|G[X_{T'}]\| \leq 2b_\F \cdot c_\F.
\]
\end{claim}
\noindent We denote $c_1:=2b_\F \cdot c_\F$.

For a node $t \in M$, let $F(t)$ be the set of those edges of $F$ that are incident to some vertex of $X_t$.
A standard hand-shaking argument shows that 
\begin{equation}\label{eq:sum-F}
\sum_{t\in M} |F(t)| \leq 2|F|.
\end{equation}
Let $T_1,\ldots,T_{p(t)}$ be the connected components of $T-t$ which contain some other nodes of $M$ and let $S_{1},\dots,S_{{q(t)}}$ be those with none.
Again, by $X_{T_i}$, resp. $X_{S_{i}}$, we denote the union of bags at the nodes of $T_i$, respectively~$S_{i}$.

Consider the forest with vertex set $M$ defined as follows: put an edge between two nodes of $M$ if there is a component of $T-M$ that neighbors both of them.
Observe that $p(t)$ is the degree of $t$ in this forest. 
Therefore, as there are at most $|M|-1$ edges in any forest on $|M|$ vertices, we infer that 
\begin{equation}\label{eq:sum-p}
\sum_{t\in M} p(t) \leq 2(|M|-1).
\end{equation}
Define $f(t) := |F(t)|+b_\F \cdot p(t)+d_\F$.
By~\eqref{eq:sum-F} and~\eqref{eq:sum-p} we conclude that 
\begin{equation}\label{eq:sum-f}
\sum_{t\in M} f(t)\ \leq\ 2|F| + b_\F\cdot 2(|M|-1) + |M|\cdot d_\F\ \leq\ d_1 \cdot |F|,\qquad\mathrm{for}\ d_1 := 2 + 8b_\F +4 d_\F.
\end{equation}
We proceed similarly as in the proof of Lemma~\ref{lem:nobouquetsLinear}. Consider an arbitrary $t\in M$.
By Corollary~\ref{corol:neatAdhesions}, at least $q(t)-(2b_\F+1)$ of the trees $S_1,\ldots,S_{q(t)}$ are connected to $t$ via a neat adhesion; let $I_1\subseteq \{1,2,\ldots,q(t)\}$ be the set of their indices.
That is, for each $i\in I_1$ we have that $|\delta(X_{S_i})|\leq 2$ and $N(X_{S_i})$ is a subset of $X_t$ of size at most~$2$. 
Since there are at most $b_\F^2$ subsets of $X_t$ of size at most $2$, at least $(q(t)-2b_\F-1)/b_\F^2$ of subtrees $\{S_i\}_{i\in I_1}$ have the same neighborhood $N(X_{S_i})=U$, 
for some $U \subseteq X_t$ of size at most~$2$; let $I_2\subseteq I_1$ be the set of their indices.
By Claim~\ref{cl:rest-protr2}, for each $i\in I_2$ we have that $X_{S_i}$ induces in $G$ a subgraph with at most $c_1$ edges.
It follows that there are at most $c_2:=(2c_1+1)^{c_1+2}$ possible isomorphism types for graphs $G[X_{S_i}\cup U]$ for $i\in I_2$ (considering isomorphisms that fix $U$).
Therefore, if $(q(t)-2b_\F-1)/(b_\F^2\cdot c_2) \geq f(t)$,
then there is a subset $I_3\subseteq I_2$ of size at least $f(t)$ for which sets $\{X_{S_i}\}_{i\in I_3}$ are elements of a single bouquet $\mathcal{X}$ attached to~$U$.

Let $\Delta\subseteq E(G)$ be the set comprising of $F(t)$ and the adhesions corresponding to the edges connecting subtrees $T_1,\dots,T_{p(t)}$ with $t$ in~$T$.
Then $\Delta$ separates in $G$ the vertices of 
$$Z:=X_t \cup X_{\Ss_1} \cup \ldots \cup X_{\Ss_{q(t)}}$$ 
from the rest of the graph; that is, all the edges between $Z$ and $V(G)\setminus Z$ are contained in $\Delta$.
Since $\Delta$ contains $F(t)$ and none of the sets $X_{\Ss_i}$ is incident to any edge of $F$, we infer that $G[Z]-\Delta$ is $\F$-free.
Hence all the components of $G-\Delta$ that contain some vertex of the bouquet $\mathcal{X}$ are $\F$-free.
By applying Lemma~\ref{lem:bouquetBound}, we infer that either 
\begin{equation}\label{eq:small-boq}
|\mathcal{X}| \leq d_\F + |\Delta| \leq d_\F + |F(t)| + b_\F \cdot p(t)=f(t),
\end{equation}
or all except $f(t)$ elements of the bouquet can be deleted to obtain a strictly smaller subgraph $G'\subsetneq G$ with $\OPT(G')=\OPT(G)$. 
Hence, if~\eqref{eq:small-boq} does not hold, then the algorithm can output $G'$ and terminate.
We can thus conclude the proof, unless for all $t\in M$ the following holds:
\begin{equation*}
(q(t)-2b_\F-1)/(b_\F^2\cdot c_2) \leq f(t),
\end{equation*}
or equivalently
\begin{equation}\label{eq:bnd-q}
q(t) \leq f(t) \cdot c_2 \cdot b_\F^2 + 2b_\F+1.
\end{equation}
We henceforth assume that this is the case.

Similarly, if there are more than $f(t)$ edges with the same pair of endpoints $U \subseteq X_t$, by Lemma~\ref{lem:thetaBound} all but $f(t)$ of them can be deleted to obtain 
a strictly smaller subgraph $G'\subsetneq G$ with $\OPT(G')=\OPT(G)$.
We can thus conclude the proof unless, for each $t\in M$, there is no group of more than $f(t)$ parallel edges connecting the same pair of endpoints in $X_t$.
We henceforth assume that the latter alternative is the case.
Since $|X_t|\leq \width'(\T)\leq b_\F$, we have the following for each $t\in M$:
\begin{equation}\label{eq:bnd-bag}
\|G[X_t]\| \leq b_\F^2 \cdot f(t).
\end{equation}

We proceed with analyzing the size of the instance, with the goal of showing that it is bounded linearly in $|F|$. By~\eqref{eq:sum-f} and~\eqref{eq:bnd-q} we have
\begin{equation*}
\sum_{t\in M} q(t)\ \leq\ c_2 \cdot b_\F^2 \cdot d_1 \cdot |F| + |M| \cdot (2b_\F+1)\ \leq\ c_3 \cdot |F|,\qquad \mathrm{for}\ c_3:=c_2 \cdot b_\F^2 \cdot d_1 + 4(2b_\F+1).
\end{equation*}
That is, the total number of components of $T-M$ with exactly one neighbor in $T$ is at most $c_3 \cdot |F|$.
As we argued before, the number of components of $T-M$ with exactly two neighbors in $T$ is at most $|M|-1$.
Hence, 
\begin{equation}\label{eq:cc-M}
T-M\textrm{ has at most }c_3\cdot |F|+|M|-1\textrm{ connected components in total.}
\end{equation}

We now examine the edges of $G$. Every edge of $G$ is either: 
\begin{enumerate}[(i)]
\item in $F$, or
\item in $G[X_t]$ for some $t\in M$, or
\item in $G[X_{T'}]$ for a component $T'$ of $T-M$, or
\item in an adhesion corresponding to an edge of $T$ connecting a connected component of $T-M$ with a node of $M$.
\end{enumerate}
The total number of edges of each of these types is respectively bounded by: 
\begin{enumerate}[(i)]
\item $|F|$, 
\item $b_\F^2 \cdot d_1 \cdot |F|$ (by~\eqref{eq:sum-f} and~\eqref{eq:bnd-bag}),
\item $c_1 \cdot (c_3 \cdot |F| + |M|-1)$ (by Claim~\ref{cl:rest-protr2} and~\eqref{eq:cc-M}), and 
\item $2b_\F \cdot (c_3 \cdot |F| + |M|-1)$ (by $\width'(\T)\leq b_\F$,~\eqref{eq:cc-M}, and the fact that each component of $T-M$ neighbors at most two nodes of $M$).
\end{enumerate}
Therefore, we conclude that 
\begin{equation*}
\|G\|\leq c \cdot |F|\qquad\mathrm{for}\ c:=1+(c_1+2b_\F)\cdot(c_3 + 4) + b_\F^2 \cdot d_1
\end{equation*}
as required.
\end{proof}

We are ready to conclude the description of our kernelization algorithm, that is, to prove Theorem~\ref{thm:main-ker}. For convenience, we recall its statement and adjust it to the current notation.

\begin{theorem}[Theorem~\ref{thm:main-ker}, reformulated]\label{thm:main-ker-recall}
There is an algorithm that, given an instance $(G,k)$ of {\sc{$\F$-Immersion Deletion}}, runs in time $\Oh(\|G\|^4 \log \|G\| \cdot |G|^3)$ and either correctly concludes that $(G,k)$ is a NO-instance, or outputs an equivalent instance $(G',k)$ such that 
$G'$ has at most $c_{\mathrm{ker}}\cdot k$ edges, where $c_{\mathrm{ker}}$ is a constant depending on $\F$ only.
\end{theorem}
\begin{proof}
We fix the constant $c_{\mathrm{ker}}$ as
$$c_{\mathrm{ker}}:=c_{\mathrm{apx}}\cdot c,$$
where $c$ is the constant given by Lemma~\ref{lem:reduceFromApprox}.
Let $(G,k)$ be the input instance. 

We first apply the algorithm of Lemma~\ref{lem:exhaustProtrusions}, which runs in time $\Oh(\|G\|^2 \log \|G\| \cdot |G|^3)$
and yields a new graph $G'$ such that $\|G'\|\leq \|G\|$, $\OPT(G')=\OPT(G)$ and each connected component of $G'$ has no excessive protrusion.
From now on, we work on the graph $G'$ instead of $G$.
If we already have that $\|G'\|\leq c\cdot k$, then we can simply output $(G',k)$, so assume that this is not the case.

Apply the approximation algorithm of Lemma~\ref{thm:approx} to $G'$, yielding in time $\Oh(\|G\|^3 \log \|G\| \cdot |G|^3)$ a subset of edges $F$ such that $G'-F$ is $\F$-free and $|F|\leq c_{\mathrm{apx}}\cdot \OPT(G')$.
If $|F|>c_{\mathrm{apx}}\cdot k$, then we can infer that $\OPT(G)=\OPT(G')>k$ and thus we terminate the algorithm by concluding that $(G,k)$ is a NO-instance.
Hence, assume otherwise, that $|F|\leq c_{\mathrm{apx}}\cdot k$.
Because $\|G'\|>c_{\mathrm{apx}}\cdot c\cdot k$, we have that 
\begin{equation}\label{eq:sol-density}
\|G'\|/|F|>c.
\end{equation}

For each connected component $H$ of $G'$, let $F_H=F\cap E(H)$.
Obviously $H-F_H$ is $\F$-free, as it is an induced subgraph of $G'-F$.
By~\eqref{eq:sol-density}, there exists some connected component $H$ of $G'$ for which $\|H\|/|F_{H}|>c$, that is, $\|H\|>c\cdot |F_{H}|$.
Therefore, as $H$ is connected and has no excessive protrusions (as a connected component of $G'$), we can apply the algorithm of Lemma~\ref{lem:reduceFromApprox} to $H$.
This application takes $\Oh(\|H\| \cdot |H|^2)$ time and outputs a subgraph $H'$ of $H$ with $\|H'\|<\|H\|$ and $\OPT(H')=\OPT(H)$.
We can now replace $H$ with $H'$ in $G'$, thus obtaining a new graph $G''$, and restart the whole algorithm on the instance $(G'',k)$.
Since $\OPT(H')=\OPT(H)$ and every graph of $\F$ is connected, it follows that also $\OPT(G')=\OPT(G)$ and, hence, the instance $(G'',k)$ is equivalent to $(G,k)$.
Also, $\|H'\|<\|H\|$ implies $\|G''\|<\|G'\|\leq \|G\|$, so the number of edges is strictly smaller in the instance $(G'',k)$ that in the original instance $(G,k)$.

We conclude that the algorithm will either terminate by concluding that $(G,k)$ is a NO instance, or it will output $(G',k)$, provided $\|G'\|\leq c_{\mathrm{ker}}\cdot k$, or
it will restart on an equivalent instance $(G'',k)$ with $\|G''\|<\|G\|$. 
The number of iterations is bounded by the number of edges in the original graph and each iteration takes $\Oh(\|G\|^3 \log \|G\| \cdot |G|^3)$ time, so the running time bound follows.
\end{proof}

\section{Bounding the size of the obstructions}
\label{obstructions}

In order to prove the second part of Theorem~\ref{thm:main-ker} it is enough to prove that, in the statement of the first part, 
 the graph of the equivalent instance $(G',k)$ is an immersion of $G$.  To see this, assume that  $H\in {\cal O}_{k}^{\rm im}={\bf obs}_{\rm im}({\cal G}^{\rm im}_{k,{\cal F}})$. Clearly, $(H,k)$ is a NO-instance of  {\sc{$\F$-Immersion Deletion}}. If we run the kernelization algorithm on $(H,k)$ the result should be a  NO-instance $(H',k)$ where $H'$ is an immersion of $H$.
As $H$ is an immersion obstruction of ${\cal G}^{\rm im}_{k,{\cal F}}$, for every proper immersion of $H$, the pair $(H',k)$
should be a YES-instance. Therefore $H'=H$ and, according to the first statement of Theorem~\ref{thm:main-ker}, $H$ has  a linear, on $k$, number of edges.\medskip

It remains now to modify the kernelization algorithm of Theorem~\ref{thm:main-ker} so that, when 
it runs with input $(G,k)$, it outputs 
a pair $(G',k)$ where $G'$ is an immersion of $G$. Recall that the algorithm, during its execution, either 
applies replacements of replaceable protrusions 
(i.e., $2b_{\cal F}$-protrusions with more than $c_{\cal F}$ edges)  
with smaller ones (Lemma~\ref{lem:basicReplacer}), 
or removes edges from thetas and buckets (Lemma~\ref{lem:reduceFromApprox}).  
Therefore we need to modify the protrusion replacement in Lemma~\ref{lem:basicReplacer} so 
that  $G'$ is an immersion of $G$.  Before we explain this modification,  we first need  some 
definitions.\smallskip

Given two $r$-boundaried graphs $\Gf=(G,(u_1,\ldots,u_r))$
\textrm{and} $\Hf=(H,(v_1,\ldots,v_r)),$ we say that $\Hf$ is a {\em rooted immersion}
of $\Hf$ if $H$ is an immersion of $G$ where the corresponding mapping $\mu_{V}$ 
maps $v_{i}$ to $u_{i}$ for every $i\in\{1,\ldots,r\}$. 
We next argue  that for every $r$, rooted $r$-boundaried graphs are well-quasi-ordered with respect to the rooted-immersion relation. Indeed, Robertson and Seymour proved in~\cite{RobertsonS10} that the set of all colored (by a bounded number of colors) graphs is well-quasi-ordered  with respect to the colored immersion relation (here the function $\mu_{V}$ should additionally map vertices to vertices of the same color). 
By
assigning to the boundaries of the $r$-boundaried graphs $r$ different colors and using one more color for their non-boundary vertices, we deduce that $r$-boundaried graphs are well-quasi-ordered under  the rooted immersion relation.\smallskip

Let ${\cal B}_{r}$ be the set of all $r$-boundaried graphs.
The set ${\cal B}_{r}$ has a partition ${\cal C}^{(r)}=\{{\cal C}_{1}^{(r)},\ldots,{\cal C}_{q_{r}}^{(r)}\}$ such that 
two $r$-boundaried graphs belong in the same set if and only if they have the same signature.
According to Corollary~\ref{cor:num-sig}, $q_{r}\leq 2^{2^{2^{\poly(r,\MX)}}}$. 
For every $i\in\{1,\ldots,q_{r}\}$,  let $\overline{{\cal C}}_{i}$ be the set of rooted immersion-minimal elements of ${\cal C}_{i}$. As $r$-rooted graphs are well-quasi-ordered under rooted-immersions, 
we have that $\overline{\cal C}_{i}^{(r)}$ is a finite set. We now consider the following set $${\cal R}_{{\cal F}}=\bigcup_{r\leq 2b_{\cal F}}\bigcup_{i\in\{1,\ldots,q_{r}\}}\overline{C}_{i}^{(r)}.$$
Notice that, for each $r\in\{1,\ldots,2b_{\cal F}\}$, each $r$-boundaried graph $\Bbb{H}$  belongs in some, say ${\cal C}_{i}^{(r)}$,
of the classes in ${\cal C}^{(r)}$, therefore it should contain as a rooted immersion some of the rooted graphs in  $\overline{\cal C}_{i}^{(r)}$. Let $c_{\cal F}^*$ be the maximum number of edges 
in an $r$-boundaried graph of $R_{\cal F}$.
This implies that every  $r$-boundaried graph $\Bbb{H}$ of more than $c_{\cal F}^*$ edges  can be replaced by an equivalent (i.e., one with the same signature)  $r$-boundaried graph   $\Bbb{H}'$ that belongs in ${\cal R}_{{\cal F}}$ and  is a rooted immersion of $\Bbb{H}$. Therefore, if $(G,k)$ is an instance of  {\sc{$\F$-Immersion Deletion}}, $\Bbb{H}$ is an $2b_{\cal F}$-protrusion of $G$ of more than $c^*_{\cal F}$ edges, and $G=\Bbb{F}\oplus \Bbb{H}$, then $(G',k)$ is an instance equivalent to $(G,k)$
where $G'=\Bbb{F}\oplus \Bbb{H}'$. Moreover, as $\Bbb{H}'$ is a rooted immersion of $\Bbb{G}$,
it follows that $G'$ is an immersion of $G$ as required.\medskip

Notice that the above argumentation does not give any way to compute $c_{\cal F}^{*}$ as 
the, inherently non-constructive, proof in~\cite{RobertsonS10} does not provide any way to compute a  bound to the size of the graphs in $\overline{\cal C}_{i}^{(r)}$ (it only says that $|\overline{\cal C}_{i}^{(r)}|$ is  a finite number). We wish to report  that it is actually possible to prove 
a constructive version of the second statement of Theorem~\ref{thm:main-ker}.
This proof is postponed in later versions of this paper, as it resides on results 
that are currently under developement.

\section{Conclusions}\label{sec:conc}
In this work we have proved that the protrusion machinery, introduced in~\cite{BodlaenderFLPST09,BodlaenderFLPST09meta,FominLMS12}, can be applied to immersion-related problems in a similar manner as to minor-related problems.
In particular, we have given a constant-factor approximation algorithm and a linear kernel for the {\sc{$\F$-Immersion Deletion}} problem, which on one hand mirrors and on the other surpasses
the results of Fomin et al.~\cite{FominLMS12} for {\sc{$\F$-Minor Deletion}}. 
Namely, while the exponent of the polynomial bounding the kernel size provably has to depend on the family $\F$ in the minor setting~\cite{GiannopoulouJLS15},
in the immersion setting we were able to give a linear kernel, with only the multiplicative constant depending on $\F$. We consider this apparent difference of complexity interesting and worth studying further.

The immediate next goal is to lift the technical assumption that all graphs from $\F$ are connected. 
While this assumption plays an important role in several of our proofs, we expect that it is not necessary and can be lifted using the techniques of Kim et al.~\cite{KimLPRRSS13} or
Fomin et al.~\cite{FominLMS12,abs-1204-4230}
that worked in the minor setting.
In fact, a constant-factor approximation, without the assumption on the connectivity of $\F$, can easily be obtained in the following way, as in the full version of the work of Fomin et al.~\cite{abs-1204-4230}.
Since Theorem~\ref{thm:Ffree_tctw} works just as well in the case of $\F$ containing disconnected graphs, 
a set of edges whose deletion makes a graph $\F$-free also makes it a graph tree-cut width bounded by some constant $a_\F$.
Thus, $\OPT_\F(G)$ is not smaller than the optimum size of a set of edges whose deletion turns $G$ into a graph of tree-cut width at most $a_\F$.
Tree-cut width is a graph parameter satisfying the conditions of Corollary~\ref{cor:par-meta}, 
hence given a graph $G$, we can construct in polynomial time a set of edges $F\subseteq E(G)$
of size at most a constant factor larger than $\OPT_\F(G)$, such that $G-F$ has tree-cut width at most $a_\F$.
A standard application of the optimization variant of Courcelle's theorem, due to Arnborg et al.~\cite{ArnborgLS91},
then gives a set $F'\subseteq E(G-F)$ such that $G-F-F'$ is $\F$-free and $|F'| =\OPT_\F(G-F) \leq \OPT_\F(G)$.
Hence by outputting $F \cup F'$, one achieves a constant factor approximation for {\sc{$\F$-Immersion Deletion}}.
For the linear kernel for {\sc{$\F$-Immersion Deletion}}, we so far do not see how to avoid the assumption on the connectivity of $\F$.

We believe that an important conceptual insight that is given by this paper is the confirmation of usefulness of the notions of tree-cut width and tree-cut decompositions.
Our work, together with a few other recent ones~\cite{Ganian0S15,MarxW14,Wollan15,KimOPST15}, shows that tree-cut width is often the right parameter to study in the 
context of problem related to immersions and edge-disjointness, and plays a similar role as treewidth for minors and vertex-disjointness.
We expect that more results of this kind will appear in future.

Clearly, the remaining insisting problem on the study of  {\sc{$\F$-Immersion Deletion}} problem 
is to consider cases where none of the graphs in ${\cal F}$ is 
planar subcubic. This comes as an analogue to instantiations of the  {\sc{$\F$-Minor Deletion}}
problem where ${\cal F}$ contain only non-planar graphs.  
In both cases the existence of a polynomial kernel can been seen as a major challenge in parameterized algorithms. Especially, for  {\sc{$\F$-Immersion Deletion}},  
further advances are necessary on the structure of graphs excluding non-planar 
or non-subcubic immersions. While some  results in this direction have appeared in~\cite{DvorakW15astru,MarxW14,Wollan15,BelmonteGLT2016thes}, it is still unclear 
whether  the current combinatorial insight can produce  general algorithmic 
results on  {\sc{$\F$-Immersion Deletion}}.

\noindent \paragraph*{Acknowledgements.} The authors wish to thank an anonymous referee for suggesting a more direct approach to finding excessive protrusions, as well as Ignasi Sau, Petr Golovach, Eun Jung Kim,
and Christophe Paul for preliminary discussions on the {\sc{$\F$-Immersion Deletion}} problem.


\end{document}